\documentclass[3p]{elsarticle} 
\AtBeginDocument{%
  }

\usepackage{xspace}
\usepackage{wrapfig}
\usepackage{bbm}
\usepackage{enumitem}
\usepackage{subcaption}
\usepackage{amsmath}
\usepackage{amsfonts}
\usepackage{xcolor}
\usepackage[normalem]{ulem}

\newtheorem{theorem}{Theorem}
\newtheorem{lemma}{Lemma}
\newdefinition{definition}{Definition}
\newtheorem{proposition}{Proposition}
\newproof{proof}{Proof}

\newcommand{\E}[1]{\mathrm{\bf E}\left[#1\right]}

\newcommand{\policyone}{FW-CAM\xspace}
\newcommand{\policytwo}{GLAM\xspace}
\newcommand{\HC}{\E{T_{\tiny\mbox{Weighted}}}}
\newcommand{\Exp}{\mbox{Exp}}

\newcommand{\kR}[2]{k_i^*\left( #1,#2 \right)}
\newcommand{\vlambda}{\boldsymbol{\lambda}}

\newcommand{\numcore}{n}
\newcommand{\numclass}{m}
\newcommand{\systemindex}{d}
\newcommand{\arrivalasym}[1]{\lambda_{#1}^{(\systemindex)}}
\newcommand{\numcoresasym}{\numcore^{(\systemindex)}}
\newcommand{\budget}{Relaxed\xspace}
\newcommand{\ourproblem}{CAM\xspace}
\newcommand{\inasym}{^{(\systemindex)}}

\newcommand{\transm}{\textbf{Q}}
\newcommand{\initp}{\textbf{p}}
\newcommand{\massv}{\textbf{z}}
\newcommand{\kdot}{\boldsymbol{k}_\pi}
\newcommand{\one}[1]{\mathbbm{1}_{#1}}
\newcommand{\widx}{\widehat{\ell}(t)}
\newcommand{\widebar}{\overline}

\newtheorem{assumption}{Assumption}

\begin{document}

\begin{frontmatter}

\title{Mean field optimal Core Allocation across Malleable jobs}

\author[cmu]{Zhouzi Li}
\ead{zhouzil@andrew.cmu.edu}

\author[cmu]{Mor Harchol-Balter}
\ead{harchol@cs.cmu.edu}

\author[unc]{Benjamin Berg}
\ead{ben@cs.unc.edu}

\address[cmu]{Computer Science Department, Carnegie Mellon University, United States}
\address[unc]{Computer Science Department, UNC Chapel Hill, United States}

\begin{abstract}%
Modern data centers and cloud computing clusters are increasingly running workloads composed of malleable jobs.  A malleable job can be parallelized across any number of cores, yet the job typically exhibits diminishing marginal returns for each additional core on which it runs.  This can be seen in the concavity of a job's speedup function, which describes the job's processing speed as a function of the number of cores on which it runs.  Examples of malleable jobs include machine learning tasks, parallelized across multiple GPUs, and large-scale database queries, parallelized across multiple servers.

Given the prevalence of malleable jobs, theoretical works have posed the problem of how to allocate a fixed number of cores across a stream of arriving malleable jobs so as to minimize the mean response time across jobs.
We refer to this as the Core Allocation to Malleable jobs (CAM) problem. Although quite a few recent papers have been written on CAM, the problem remains challenging to solve optimally.  
In particular, prior theoretical work has only solved the CAM problem in the case where all jobs follow the same speedup function, with very little progress on multiple speedup functions.

This paper takes a new approach to CAM, looking at the problem in the mean field asymptotic regime.  In this regime, we are able to optimally solve a much more general instance of CAM.  We allow for any number of job classes, each with an arbitrary concave speedup function and weight.
Furthermore, unlike prior work, we allow for general independent inter-arrival times and phase-type job sizes.

We derive two distinct mean field optimal policies, FW-CAM and GLAM.
FW-CAM demonstrates a new intuition: in the mean field regime, individual job sizes are not relevant in finding an optimal policy.
GLAM (Gittins-Like Allocation for Malleable jobs) is asymptotically optimal under certain common assumptions, and also serves as a good heuristic \emph{outside} of the asymptotic regime by adapting dynamically to the current number of jobs in the system.
Notably, none of the policies previously proposed in the literature are mean field optimal when jobs follow different speedup functions.
\end{abstract}%



\end{frontmatter}

\section{Introduction}

Modern data centers and cloud computing clusters are increasingly tasked with hosting workloads composed of \textit{malleable} jobs, such as machine learning training tasks~\cite{jayaram2023sia, qiao2021pollux}, large-scale database queries~\cite{leis2014morsel,wagner2021self}, and data center computing tasks~\cite{tumanov2016tetrisched,delgado2018kairos}. Unlike traditional single-server jobs, malleable jobs can utilize multiple processor cores to accelerate their execution. Furthermore, unlike multi-server jobs that require a specific number of cores to run, malleable jobs do not have hard constraints on the degree of parallelism; a malleable job can execute on {\em any} number of cores, and can change its degree of parallelism as it runs. Given a stream of malleable jobs that arrive to a system over time, a system operator must allocate a fixed pool of $n$ cores across the arriving jobs to optimize system performance (e.g., the mean response time, where a job's response time is the time from when it arrives until it completes). Designing an efficient \textit{Core Allocation Policy} for this setting remains a significant open challenge.

Motivated by these applications, this paper studies the problem of \textit{Core Allocation to Malleable jobs} (which we call \ourproblem). In this problem, we are given $n$ cores and a stream of malleable jobs. The goal is to, at all times, determine how many cores to assign to each job -- the number of cores assigned to a job can change over time.   The objective is to minimize the (weighted) mean response time (see Section~\ref{sec:setting}). 

\subsection{Speedup functions and their effect on running time}
\begin{figure}[t]
    \centering
    \begin{subfigure}[t]{0.35\textwidth}
        \centering
        \includegraphics[width=0.8\textwidth]{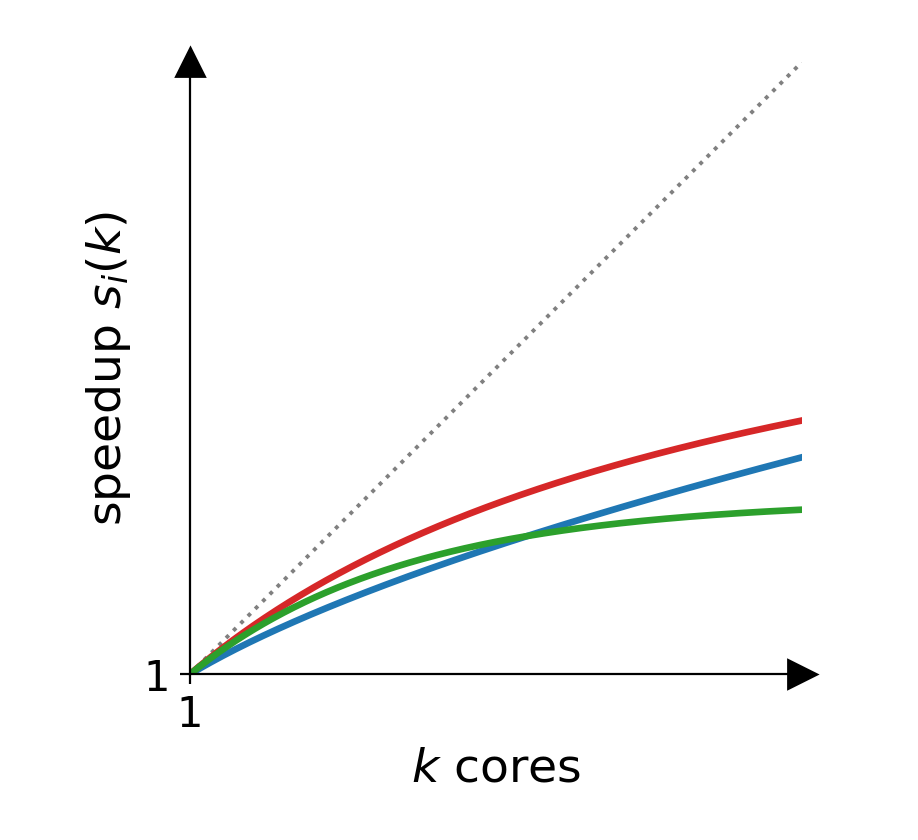}
        \caption{Speedup functions $s_i(k)$.}
        \label{fig:speedup-shapes}
    \end{subfigure}\hspace{2em}
    \begin{subfigure}[t]{0.35\textwidth}
        \centering
        \includegraphics[width=0.8\textwidth]{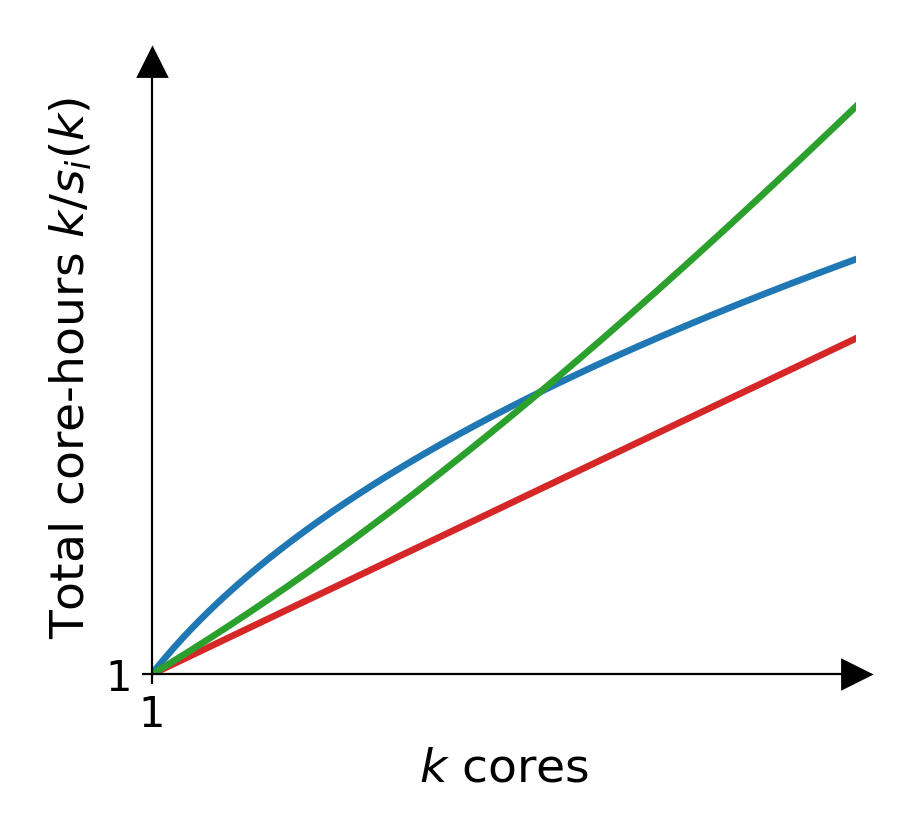}
        \caption{Total core-hours $k/s_i(k)$.}
        \label{fig:k-over-s}
    \end{subfigure}
    \caption{Examples of speedup functions and the corresponding total core-hours (assuming job size is~$1$). Speedup functions are typically sublinear and concave; $k/s(k)$ is increasing and at least~$1$, but it is not restricted to being convex or concave.}
    \label{fig:speedups}
\end{figure}
A defining characteristic of this problem is the diminishing marginal returns associated with parallel execution.
A job's \emph{speedup function}, $s(k)$, describes how much faster a job runs on $k$ cores compared to a single core.
Job speedup functions are typically concave and sublinear (see Figure~\ref{fig:speedup-shapes} for some examples).  Hence running a job on 4 cores may only increase its speed by a factor of 2.

The \ourproblem problem is complex because there are multiple classes of jobs, where every job class is associated with a different speedup function.  Specifically, $s_i(k)$ denotes the speedup function of a class-$i$ job (see Section~\ref{sec:setting} for the formal definition).   Additionally, every class of job can have a different {\em size} distribution.  The {\em size} of a job is the inherent work associated with the job, which we measure as the job's running time on a single core.
The {\em running time} of a job of size $x$ and class $i$ on $k$ cores is given by $\frac{x}{s_i(k)}$, as explained in Section~\ref{sec:setting}.


In the \ourproblem problem, as more cores are allocated to a single job, the overall system efficiency (defined as the total inherent work completed per unit of core time) decreases.
This efficiency loss is captured by the distinction between \emph{system load} and \emph{effective load}, defined as follows.

For an $\numclass$-class GI/GI/$\numcore$ system ($\numcore$ cores), the  {\em system load}, $\rho$, is defined as 
$$\mbox{system load} = \rho = \frac{\sum_{i=1}^{\numclass}\lambda_i \E{X_i}}{\numcore} < 1,$$
where $\lambda_i$ is the arrival rate of class-$i$ jobs, $X_i$ is a random variable denoting the job size of a class-$i$ job, and $\numcore$ is the number of cores.
While the system load can be defined exactly the same way for  the \ourproblem problem, the system load alone no longer captures the queueing behavior of the system.  
This is because parallelization increases the "effective work" of a job due to sublinear speedup (see Figure~\ref{fig:k-over-s} for an example).
For instance, if a job runs on 4 cores but achieves only a $2\times$ speedup, it consumes twice the core-seconds compared to the case when it is run on 1 core, or, equivalently,  the job's ``effective work'' is double its inherent work. 
Consequently, the \textit{effective load} on the system is strictly greater than the inherent system load, and the system may become unstable even when the system load is smaller than 1.
Thus, the queueing behavior of a system is truly determined by the {\em effective load}, which depends on both the system load and the core allocation policy. 

\subsection{A Brief History of the CAM Problem}
\label{s:brief_history}

A variety of work from the systems community considers scheduling malleable jobs \cite{lin2018model,qiao2021pollux,jayaram2023sia,tumanov2016tetrisched,wagner2021self}.
However, all of these systems rely on simple heuristic policies with no theoretical guarantees.

The \ourproblem problem has been studied by the performance modeling community over the past decade.  
Much of this work assumes exponential job sizes and Poisson arrivals \cite{berg2017towards,berg2021optimal,berg2024asymptotically}.

There are several papers devoted to the case where all jobs follow a {\em single} speedup function~\cite{berg2017towards,berg2021optimal, berg2021hesrpt}.  
Even the single speedup function problem is difficult since one needs to balance two competing goals:  (1) giving more cores to jobs with short remaining times, and (2) maximizing system efficiency. The optimal balance is achieved in the heSRPT paper \cite{berg2021hesrpt}, which assumes all jobs are present at the start.
The problem is still open when jobs arrive over time.

\begin{figure}[t]
    \centering
   \includegraphics[width=0.45\textwidth]{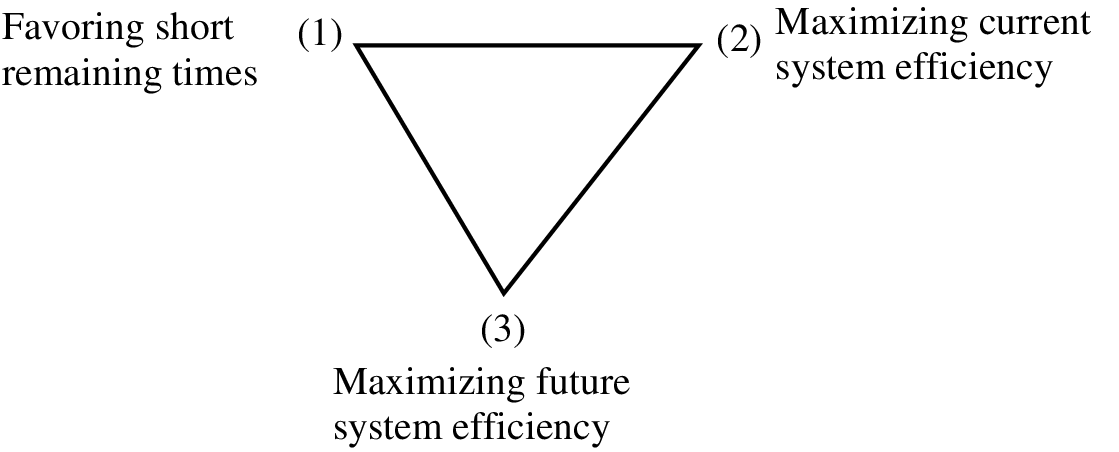}
    \caption{Three competing priorities when dealing with multiple speedup functions.}
    \label{fig:triangle}
\end{figure}

Once one moves to multiple speedup functions, the problem becomes even harder. Berg et al. \cite{berg2017towards} suggest that, in general, one must trade off three goals, shown in Figure~\ref{fig:triangle}: (1) favoring jobs with short remaining times, (2) maximizing current system efficiency, and (3) maximizing future system efficiency (specifically "saving up" some highly parallelizable jobs for times when there are only a few jobs in the system). 
Prior works have looked at special cases that reduce to only two out of the three tradeoffs.
For example, \cite{berg2020optimal,berg2022case} consider the case of two speedup functions, where the shorter jobs are also the less parallelizable jobs, aligning corners (1) and (3).  
As another example, \cite{berg2024asymptotically} assumes all speedup functions follow a {\em rooftop} shape: class-$i$ jobs are perfectly parallelizable up to $k_i$ cores, but get no further benefit from $> k_i$ cores.  
There, it is easy to satisfy corner (2), by avoiding the flat parts of the speedup functions, but trading off between (1) and (3) still requires looking at the case of  $\rho \to 1$.


In summary, the \ourproblem problem with multiple speedup functions is still largely wide open.

\subsection{Asymptotic Regimes}
\label{sec:effective}
In Section~\ref{s:brief_history}, we saw that an optimal solution to the \ourproblem problem is elusive, since it must trade off several competing goals.
In fact, we note here that  the \ourproblem problem is strictly more challenging than the classical M/G/$\numcore$ scheduling problem (where jobs are not parallelizable and each job occupies a single server).  The optimal scheduling policy for the M/G/$\numcore$ remains open except under asymptotic regimes (see Section~\ref{sec:prior} for more discussion). Hence, in order to build intuition for a good allocation policy for \ourproblem, we too turn to asymptotic analysis.

We limit our attention to the wide class of regimes where both the arrival rate and number of cores go to infinity while the job sizes stay fixed.  In traditional M/G/$\numcore$ systems, the regimes are differentiated by how the arrival rate scales compared to the number of cores.  For example, in the {\em mean field} regime~\cite{Dai_2010,gast2017refined,gast2012mean,gast2010mean}, the arrival rate scales proportionately to the number of cores; hence, the system load $\rho$ is held fixed. This is a good representation for real-world large-scale systems that operate under reasonable (not too high) loads. By contrast, in non-mean-field regimes, the arrival rate scales faster than the number of cores, causing the system load to approach 1 (\cite{halfin1981heavy,gupta2019load}). 

In the \ourproblem problem, running jobs on more than 1 core causes the effective load to be higher than the system load.  Hence for the \ourproblem problem, regimes where the system load $\rho$ approaches $1$ are not interesting because they do not allow jobs to be parallelized.  
Specifically, given jobs that follow strictly sublinear speedup functions, if any constant fraction of jobs run on more than 1 core, the effective load is a constant factor larger than the system load (which is already approaching 1), resulting in an  effective load greater than $1$ and system instability. We formalize this intuition in Proposition~\ref{lemma: app asym}.
This observation compels us to focus on the \textit{mean field} asymptotic regime (see Section \ref{sec:setting}), where the system load remains constant while the number of cores scales. While we will show that there is no conventional queuing under the mean field regime, jobs still contend with one another for cores.
Hence a scheduling policy must decide how to allocate cores to the jobs to optimize performance.

\subsection{Our Approaches and Impacts}

We derive two distinct mean-field optimal policies, \textbf{\policyone} and \textbf{\policytwo}, using fundamentally different methodologies.


\noindent\textbf{\policyone (Definition~\ref{def:FW CAM}):}
In prior work, the tradeoffs shown in Figure \ref{fig:triangle} are difficult to balance because the importance of each conflicting goal varies with the system state.
Our key observation is that in the mean field limit, the \emph{average} job arrival rates, job sizes, and speedup functions for each class are the only quantities that matter.
This allows us to derive a mean field optimal policy, \policyone, as a function of these averages. 

While conceptually simple, \policyone reveals a critical insight: a mean field optimal policy should treat all jobs in the same class (i.e., having the same speedup function) equally, allocating cores independently of the job's remaining size. That is, every job class is associated with a \emph{fixed width} (FW).
This contradicts the classical scheduling intuition of "prioritizing short jobs" (e.g., \cite{berg2021hesrpt,berg2024asymptotically}).
In fact, any policy that does prioritize short jobs is {\bf not} mean field optimal (see Proposition~\ref{prop:app no job size}). 


\noindent\textbf{\policytwo (Definition~\ref{def:policytwo}):} 
While \policyone successfully exploits the averaging effect of the mean-field regime, its allocations are static over time.
As a result, we find that \policyone's performance degrades outside the asymptotic limit (if the number of cores, $\numcore$, is small), raising questions about the practicality of using \policyone in the real world.
This motivates us to introduce a second policy \policytwo (Gittins-Like Allocation policy for Malleable jobs) that is both mean-field optimal \emph{and} adapts to the state of the system over time.

\policytwo uses a heuristic based on the Gittins index \cite{gittins1979bandit} to form a dynamic, mean-field optimal policy.
While the Gittins index has classically been used for scheduling problems with non-parallelizable jobs, we use the approach to handle malleable jobs.
Roughly speaking, \policytwo computes every job's marginal improvement in response time from running on an additional core.
Running on an additional core also increases the number of core-hours the job will use (see Figure~\ref{fig:k-over-s}), so \policytwo also computes the marginal increase in resource usage from an additional core.
\policytwo favors jobs with a higher ratio of marginal response time improvement to marginal resource usage increase.

Notably, \policytwo uses a metric that does not cleanly align with any of the tradeoffs described in Figure \ref{fig:triangle}.
In particular, it is not maximizing the classical definition of system efficiency.
In fact, we can show that the GREEDY policy \cite{berg2017towards} that maximizes system efficiency is not mean field optimal (see Section \ref{sec:simulation} and ~\ref{app:greedy not mf}).
Furthermore, \policytwo (like \policyone) does not favor short jobs in mean field.

\subsection{Our Contributions}

The contributions of this paper are as follows: 
\begin{itemize}[leftmargin=*]
    \item We derive two mean field optimal policies, FW-CAM and GLAM, for the \ourproblem problem under general settings. In particular, we allow for generally-distributed independent inter-arrival times, phase-type job sizes, multiple job classes with distinct speedup functions, and any weighted mean response time metric. This setting greatly generalizes the prior work (see Section~\ref{s:brief_history}) which requires very restrictive modeling assumptions.
    \item (Section~\ref{sec: FW Cap}) Our first policy, \policyone, leverages the concept of effective load to derive a mean field optimal policy, where all jobs in the same class get the same allocation. 
    \item (Section~\ref{sec: sec 6}) Our second policy, \policytwo, serves as a bridge between theory and practice. It is mean field optimal under certain assumptions and demonstrates superior empirical performance, converging rapidly to the optimal limit. It serves as a strong heuristic for general \ourproblem settings.
    \item Prior work on the \ourproblem problem has been largely limited to only {\em one} speedup function or a narrow class of speedup functions (see Section~\ref{s:brief_history}).
    Even for the setting of a single speedup function, our analysis provides novel insights challenging the conventional wisdom of prioritizing short jobs. 
\end{itemize}

\section{Prior Work}
\label{sec:prior}

In Section~\ref{s:brief_history}, we already reviewed the prior work on the \ourproblem problem.  In this section, we discuss prior work on scheduling parallelizable jobs under different models of parallelism that are complementary to the \ourproblem problem.


\subsection{Scheduling problems relevant to, but different from \ourproblem}


Even under traditional multi-server scheduling, where jobs are not parallelizable and each job occupies only a single server, the optimal policy is still widely open. For M/G/$n$ systems, the only known optimality result applies in the conventional heavy traffic regime~\cite{grosof2019srpt}. However, it has recently been shown that SRPT-$n$ is {\em not} optimal under non-limiting load~\cite{grosof2024bounds}, and this line of work is still active.

Another active line of work on multi-server scheduling involves {\em multiserver jobs} (e.g. ~\cite{grosof2022optimal,chen2025improving}).  These jobs are parallelizable, but not malleable.  Specifically, in~\cite{grosof2022optimal,chen2025improving} each job requires a fixed given number of cores specific to that job. It is hard to translate intuitions learned from the multiserver job setting to the \ourproblem setting where jobs are malleable and can run on {\em any }number of cores.

Additionally, there has been significant work on \emph{fork-join parallelism} \cite{kim1989analysis,baccelli1985simple}, where a single job is composed of a fixed number of tasks that queue and run independently.  
A single job is completed when all of its tasks are completed.
Analysis of fork-join queueing systems is notoriously difficult, and known results are generally limited to approximations \cite{nelson2002approximate} or asymptotic analysis {\cite{wang2019delay,zubeldia2020delay}}.

Each of the above models is important for some range of systems.
However, none of these models capture the many systems
where jobs are malleable, and the scheduler chooses the level of parallelism for each job.
We do note that many of the considerations and tradeoffs that arise in the \ourproblem problem can be seen in some form for different parallel scheduling models.
However, none of the results from these other models directly generalize to provide results in the \ourproblem model.

\subsection{The Gittins Index}
In the latter half of this paper, we use techniques related to the Gittins index \cite{gittins1979bandit} to derive a mean field optimal policy for the \ourproblem problem.
The Gittins index has been used for optimal scheduling in an M/G/1 queue \cite{aalto2009gittins,scully2021gittins,aalto2023minimizing} where job sizes are unknown, but the job size distribution is known.
The Gittins index was further used in an M/G/n setting, where it was shown to minimize mean response time in the heavy-traffic limit ($\rho \to 1$)\cite{grosof2019srpt}.  Importantly, the Gittins index is not optimal outside of heavy-traffic for the M/G/n.  
Our work is the first to use a Gittins index approach for scheduling jobs governed by speedup functions. 
Our model considers jobs of known size, while classical scheduling applications of the Gittins index focus on jobs with unknown sizes.
Hence, our work should be viewed as a generalization of SRPT-like policies rather than a generalization of prior Gittins policies for the M/G/1 and M/G/n.

\section{Core Allocation to Malleable jobs: Our Problem Setting}
\label{sec:setting}

We consider an $\numclass$-class GI/GI/$\numcore$ system with $\numcore$ homogeneous cores.

\subsection{Job Classes}
There are $\numclass$ classes of jobs.  
Each class has its own arrival process, job size distribution, speedup function and weight, all to be defined below.

\paragraph{Malleable jobs} We consider scheduling \emph{malleable jobs} that can run on any number of cores and can change the number of cores they use over time.
Hence, an arriving job can immediately run on any available cores and can then grow and shrink over time to accommodate future arrivals and departures.
We assume that there is no overhead when preempting or changing a job's allocation.
This simplification is also assumed in most prior theory works, e.g., \cite{berg2020optimal, berg2024asymptotically,berg2022case,Arpan2024,berg2021hesrpt}. 

\paragraph{Arrival process}
The arrival process of class $i$ is defined by i.i.d. inter-arrival times drawn from distribution $A_i$. We assume that $A_i$ has mean $\frac{1}{\lambda_i}$. Define $\lambda:=\sum_{i = 1}^\numclass \lambda_i$ to be the total arrival rate of jobs and let $\vlambda = (\lambda_1, \lambda_2, \ldots, \lambda_m)$ be the vector of arrival rates for each job class.
Let $\mathbb{A}_i(t)$ be the number of class-$i$ jobs that have arrived by time $t$ and let $\mathbb{A}(t) = \sum_{i = 1}^m \mathbb{A}_i(t)$ denote the total number of arrivals by time $t$.

\paragraph{Job size distribution}
The {\em size} of a job is defined as the time needed to complete a job on a single core. We can think of this as the {\em inherent work} associated with the job. The sizes of class $i$ jobs are i.i.d. random variables, $X_i$, with $\E{X_i}=1/\mu_i$. Throughout, we assume job sizes are known, but surprisingly, the exact size of each job turns out to be useless in the mean field regime.
While our results on \policyone hold for general, independent job sizes, our results on \policytwo require phase-type job sizes. 

\paragraph{Speedup function}

A job's running time depends on the number of cores the job runs on and on its {\em speedup function}.
We define $s_i: \mathbb{R}^+\to \mathbb{R}^+$ to be the {\em speedup function} of class $i$ jobs. If a class $i$ job with size $x$ runs on $k$ cores, then its running time becomes
$$ \mbox{Running time on } k \mbox{ cores} = \frac{x}{s_i(k)}.$$

Note that we assume that jobs can be assigned fractional cores, but a job cannot be assigned between 0 and 1 core. \footnote{This assumption simplifies the proof and the exposition. Our model can easily be generalized to the case where jobs require a minimum of $\epsilon$ cores to run, for any $\epsilon>0$.}
We make the following mild assumptions on the speedup functions: For any job class $i$, the speedup function $s_i(\cdot)$ should fulfill the following properties:
\begin{itemize}[leftmargin=*]
\item $s_i(1)=1$.
    \item Sub-linearity: $s_i(k)\leq k$ for $k\geq1$.
    \item Monotonicity: for any $1\leq k_1<k_2$, we have that $s_i(k_1)\leq s_i(k_2)$.
    \item Concavity: for any $1\leq k_1<k_2$, we have that $s_i(k_1)/k_1 \geq s_i(k_2)/k_2$.
    \item Smoothness: $s_i'(k)$ exists for all $k>1$.
\end{itemize}

Although real-world speedup functions may not be perfectly concave, in practice, this is dealt with by using the monotonic, concave hull of an empirically measured speedup function \cite{Arpan2024}.

Our second mean-field optimal policy, \policytwo,  needs the following additional assumptions:
\begin{itemize}[leftmargin=*]
    \item Strict sub-linearity: $\lim_{\epsilon\to 0^+}\frac{s_i(1+\epsilon)-s_i(1)}{\epsilon}<1$ and $s_i(k)=o(k)$ for $k\to \infty$.
    \item Strict concavity: $s_i''(k)$ exists and $s_i''(k)<0$ for any $k>1$.
    
\end{itemize}

\paragraph{Stability Conditions}
We assume some additional conditions on the job size distribution and interarrival distribution that are necessary for system stability.
In addition to the standard assumption that system load is less than 1 (see Section \ref{sec:mf}), additional distributional assumptions are required to guarantee stability.
In particular, the stability conditions for the \ourproblem system mirror those of a GI/GI/$n$ queue.
A full treatment of these stability conditions is given in \cite{gamarnik2013steady}, and is beyond the scope of this paper.
In this paper, we assume $\E{A_i^2}<\infty$ and $\E{X_i^3}<\infty$ for all $i$, along with all other conditions on $A_i$ and $X_i$ that are needed for stability.

\subsection{Performance Metrics}

 \paragraph{Weighted mean response time}
Our goal is to minimize the weighted mean response time across jobs, where class $i$ is assigned weight $c_i$.
 Formally, let $T_i$ be the stationary response time of class-$i$ jobs with mean $\E{T_i}$. 
 The {\em weighted mean response time} of the system is given by

 \begin{equation}
     \HC:= \sum_{i=1}^\numclass \frac{\lambda_i}{\lambda} c_i\E{T_i}.
     \label{eq:HC}
 \end{equation}

\paragraph{Optimizing weighted mean response time}
Let $k_{ij}(t)$ be the number of cores allocated to the $j$th class-$i$ job at time $t$.
Let $T_{ij}$ be the response time of the $j$th class-$i$ job.
{More specifically, let $a_{ij}$ and $X_{ij}$ denote the arrival time and realized inherent work of this job, respectively. With the convention $s_i(0)=0$, its response time is defined by}
\[
{
T_{ij}:=\inf\left\{u\geq 0:\int_{a_{ij}}^{a_{ij}+u}s_i\bigl(k_{ij}(t)\bigr)\,dt\geq X_{ij}\right\}.
}
\]
{Thus, $T_{ij}$ is jointly determined by the job's inherent work and its entire allocation trajectory.}
We can then write the following optimization problem, which describes the problem of choosing job allocations to minimize the weighted mean response time across jobs in the \ourproblem problem.

\begin{equation}
        \begin{aligned}
        & \underset{k_{ij}(t)} {\text{minimize}} & \lim_{\tau \to\infty}\frac{1}{\mathbb{A}(\tau)}\sum_{i=1}^m\sum_{j=1}^{\mathbb{A}_i(\tau)} T_{ij} c_i 
        \qquad &\mbox{(weighted mean response time)}
        \\
        & \text{subject to} & \sum_{i = 1}^m\sum_{j=1}^{\mathbb{A}_i(\tau)}k_{ij}(t) \leq \numcore, \quad  \forall t\geq 0 &\qquad\mbox{(usage constraint at time $t$)}\\
        & &k_{ij}(t) \in \{0\}\cup [1,\infty) \quad \forall  t\geq 0.\\
        \end{aligned}
    \label{eq:opt:general}
\end{equation}

Our goal in this paper is to find, among policies for which the above limit exists, the policy with the smallest long-run average weighted mean response time.
We will refer back to this optimization problem frequently as a helpful starting point for deriving new scheduling policies.




\subsection{The Mean Field Regime}
\label{sec:mf}
The main results of this paper consider the {\em mean field scaling regime}.
In the mean field regime, the number of classes, job size distributions and speedup functions are held fixed, and the arrival rate is scaled proportionally with the number of cores.
Formally, a sequence of systems indexed by $\systemindex$ are considered as $d\to\infty$. Let $\lambda_i^{(\systemindex)}$ be the arrival rate of type-$i$ jobs in the $\systemindex^{th}$ system, and let  $\numcore^{(\systemindex)}$ be the number of cores in the $\systemindex^{th}$ system. Then in the mean field regime, $\arrivalasym{i}=\systemindex\cdot \lambda_i$, $\vlambda\inasym=d\cdot\vlambda$, and $\numcoresasym=\systemindex\cdot \numcore.$ 
More precisely, inter-arrival times in the $\systemindex^{th}$ system are sampled from $A_i\inasym:=\frac{A_i}{\systemindex}$. In this way, the mean inter-arrival time scales, but the coefficient of variation of the inter-arrival time distribution stays constant. 
Importantly, in the mean field regime, the {\em system load} stays constant as $d$ is increased.  Mathematically, the system load in the $d^{th}$ system is defined as $$\rho := \sum_{i=1}^{\numclass} \rho_i \quad \mbox{ where } \quad \rho_i:= \frac{1}{\numcore\inasym}  \lambda_i\inasym \cdot \E{X_i}.$$ Neither $\rho$ nor $\rho_i$ depend on $\systemindex.$
We assume throughout the paper that $\rho<1$, which is necessary for system stability.

We refer to the \ourproblem problem in the $\systemindex^{th}$ system as $\ourproblem \inasym$.  We denote the weighted mean response time of a policy $\pi$ in the $\mbox{\ourproblem}\inasym$ system by $\HC^\pi_{\ourproblem\inasym}$.
The policy $\pi$ is considered {\em mean field optimal} for the \ourproblem problem if, for any other policy $\pi'$, we have
$$\lim_{d\to\infty}\frac{\HC^\pi_{\ourproblem\inasym}}{\HC^{\pi'}_{\ourproblem\inasym}} \leq 1.$$
Here we assume the limits of both the numerator and the denominator separately exist.

\section{Preliminaries: A relaxed version of the $\ourproblem$ problem}
\label{sec:relax}
In this section, we introduce a relaxed version of the \ourproblem problem, which provides a lower bound on the optimal weighted mean response time, $\HC^*$, for the \ourproblem problem.  
In the \budget problem, we drop the hard constraint that at every moment of time only $\numcore$ cores are available. Instead, we allow the total number of cores to fluctuate at our control, so long as the time-average number of cores stays within a budget of $\numcore$. 
Similar relaxations are commonly used in queueing problems to prove asymptotic optimality~\cite{Verloop_2016,weber1990,whittle1988restless,eshwar2024online,gast2024lp,hong2024exponential}.

Specifically, consider the $j$th job of class $i$ that arrives at time $t$.  Let 
$$\widebar{k}_{ij}=\frac{\int_{t}^{t+T_{ij}} k_{ij}(s) ds}{T_{ij}}$$
be the average number of cores used to process the $j$th job of class $i$.  Looking at the optimization problem in \eqref{eq:opt:general}, we relax the constraint on the total number of cores used at time $t$ to
\begin{equation}
\frac{1}{\tau}\sum_{i = 1}^m\sum_{j=1}^{\mathbb{A}_i(\tau)}\widebar{k}_{ij} T_{ij} \leq \numcore.
\label{eq:relaxed:budget}
\end{equation}
It is straightforward to see that this new constraint is weaker than the original constraint in \eqref{eq:opt:general}, since the original constraint clearly implies that the time average core usage is below $\numcore$.  Hence, the weighted mean response time for the \budget problem is a lower bound on the optimal weighted mean response time in the \ourproblem problem, since the constraint of a fixed number of cores is relaxed.
Formally, we denote the optimal weighted mean response time for the \ourproblem problem and the \budget problem by $\HC_{\ourproblem}^*$ and  $\HC_{\budget}^*$, respectively. Then we have that 
\begin{equation}
    \HC_{\ourproblem}^* \geq \HC_{\budget}^*.
    \label{eq:budget is lower bound}
\end{equation}



We addressed the \budget problem in a prior technical report in
\cite{li2024rentgpusbudget}.
However, \cite{li2024rentgpusbudget} does not consider mean-field scaling nor does it address applications to the original CAM problem.
Instead, \cite{li2024rentgpusbudget} solves the Relaxed problem by deriving two results.
First, we showed in \cite{li2024rentgpusbudget} that the optimal policy for the Relaxed problem should never queue jobs; 
intuitively, queueing any job makes its response time worse without reducing the time-average number of cores it will eventually use.
Second, we showed that the optimal Relaxed policy should not change a job's core allocation over the course of the job's lifetime.
This is a consequence of the concavity of the job's speedup function --- a job whose allocation changes over time will run faster by receiving a constant allocation equal to its time-average core allocation.
A similar argument can be used to show that two different jobs of the same class should receive equal allocations.
Taken together, these properties imply that one can find the optimal policy for the Relaxed problem by computing the optimal number of cores, $k_i$, to allocate to every class-$i$ job.
This $k_i$ will replace the $\widebar{k}_{ij}$ in our Relaxed optimization problem, \eqref{eq:relaxed:budget}. 
The optimal Relaxed policy then allocates the appropriate number of cores to each arriving job with no queueing.

Given this structure of the optimal Relaxed policy, it is straightforward to obtain simplified expressions for both the weighted mean response time from \eqref{eq:opt:general} and the relaxed core usage constraint in \eqref{eq:relaxed:budget}.
This result is summarized in Theorem \ref{thm: optimal policy for relax}, as follows.

\begin{theorem}[Theorem 1 in \cite{li2024rentgpusbudget}]
    \label{thm: optimal policy for relax}
    The optimal policy for the \budget problem is: Whenever a class-$i$ job arrives in the system, $k_i$ cores are immediately allocated to it until it completes, where $k_i$ is the solution to the following convex optimization problem 
    \begin{equation}
        \begin{aligned}
        & \underset{k_i} {\text{minimize}} &\sum_{i=1}^m \frac{\lambda_i}{\lambda} \E{T_i} c_i=
        &  
        \sum_{i=1}^\numclass \frac{\lambda_i}{\lambda}\frac{ \E{X_i} c_i }{s_i(k_i)} \quad &\mbox{(minimize weighted mean response time)}
        \\
        & \text{subject to} & \sum_{i = 1}^m \lambda_i \E{T_i} k_i =&  
        \sum_{i=1}^\numclass\frac{\lambda_i \E{X_i} k_i}{s_i(k_i)}\leq \numcore, \quad &\mbox{(time-average budget constraint)}\\
        & 
        &  
        k_i\geq& 1.\\
        \end{aligned}
    \label{eq:online opt}
\end{equation}
The convexity of \eqref{eq:online opt} is shown in  \ref{app:cvx}.
\end{theorem}
We refer to the optimal policy for the Relaxed problem as the \emph{Relaxed policy}.  We note that the Relaxed policy is an example of a \emph{fixed-width policy} (FW for short) because every job of class $i$ gets the same number of cores, $k_i$.

We note that, even in the \budget model, the above results are counterintuitive.
Because the optimal policy does not queue jobs, the typical benefits of prioritizing short-running jobs ahead of long-running jobs disappear.
Specifically, given two jobs with the same speedup function but vastly different job sizes, the optimal policy will allocate the same number of cores to each job, contrasting with the tradeoffs in Figure~\ref{fig:triangle}.


We now apply Theorem \ref{thm: optimal policy for relax} in the mean field regime where $\arrivalasym{i} = \systemindex\cdot \lambda_i$ and $\numcore\inasym=\systemindex\cdot \numcore.$
Importantly, the optimization problem in $\eqref{eq:online opt}$ when written for the $\systemindex$th system, is independent of $d$.
Hence, \eqref{eq:online opt} is also a lower bound on the weighted mean response time of the mean-field optimal policy for \ourproblem.
%

\section{The first Mean Field optimal policy: \policyone}
\label{sec: FW Cap}

In this section, we present the first mean field optimal policy for the \ourproblem problem, \policyone. In Section~\ref{sec: FW CAM intuition} we first develop the intuition behind our policy. Next in Section~\ref{sec: FW CAM def} we present our \policyone policy. Finally in Section~\ref{sec: FW CAM proof} we present the proof of its mean field optimality.

\subsection{Intuition behind \policyone}
\label{sec: FW CAM intuition}

The goal of this section is to use the Relaxed policy (see Section~\ref{sec:relax}) to derive an asymptotically optimal policy for the \ourproblem problem, \policyone.
There are two issues with doing this.

First, the Relaxed policy is allowed to use more than $\numcore$ cores during some periods of time.  In the \ourproblem system, jobs would be forced to queue during these periods.

Second, if the Relaxed policy is directly used in the \ourproblem system, allocating $k_i$ cores to every class-$i$ job, the resulting \ourproblem system will be unstable.  
To see this, consider the \emph{effective load} of the Relaxed policy in the \ourproblem system.
When each class-$i$ job runs on $k_i$ cores, it consumes $\E{T_i} k_i = \frac{\E{X_i}k_i}{s_i(k_i)}$ amount of core time (i.e., CPU-seconds) in expectation. 
Hence, the total effective load of the Relaxed policy in the \ourproblem system is 
\[\text{Total Effective Load}= r = \sum_{i=1}^\numclass r_i  \qquad \mbox{ where } \qquad r_i =\frac{\lambda_i}{\numcore} \cdot\frac{\E{X_i}k_i}{s_i(k_i)} \]
Here, $r_i$ represents the effective load contributed by class $i$ jobs.

Because the budget constraint in \eqref{eq:online opt} will be tight when deriving the optimal Relaxed policy, the effective load of the optimal Relaxed policy will be 1 in the \ourproblem system, leading to system instability.

The above issues motivate the rest of the section.
Our goal is to reduce the effective load of the Relaxed policy in the \ourproblem system such that queueing disappears.
We then argue that, in the mean field scaling regime, we can make this change without significantly reducing the number of cores allocated to each job, and hence without impacting performance. 
The resulting policy will be mean-field optimal because it will perform similarly to the Relaxed policy in the Relaxed system, whose weighted mean response time is a lower bound on the weighted mean response time of the \ourproblem system.

\subsection{Defining \policyone}
We begin by formally defining \policyone, which will adapt the Relaxed policy to the CAM system.
Roughly speaking, under \policyone, we set aside a small number of cores (on the order of $n^{0.75}$).
With the remaining $n\inasym - \left(n\inasym\right)^{0.75}$ cores, we solve the optimization problem \eqref{eq:online opt} to find the optimal Relaxed policy that allocates the $n\inasym-\left(n\inasym\right)^{.75}$ cores on average. 
We then apply the resulting Relaxed policy in the original \ourproblem system.
The resulting policy uses only $n\inasym - \left(n\inasym\right)^{0.75}$ out of the $n$ total cores on average.
Hence, the effective load of \policyone is less than 1 and the system is stable.
As $d$ increases, the effective load of \policyone in the $\systemindex^{th}$ system scales as roughly $1-\systemindex^{-0.25}$. 
This scaling regime lies within the Sub-Halfin-Whitt regime\footnote{The Sub-Halfin-Whitt regime refers to systems where load scales as  $1-\omega(\systemindex^{-0.5})$. We enforce a stronger scaling regime ($1-\omega(\systemindex^{-0.25})$) in order to apply the results in \cite{li2025simple}.}.  
We will show that the mean queueing time in the CAM system goes to $0$ in this scaling regime. 

We start by defining \policyone for an arbitrary \ourproblem system with a vector of arrival rates $\vlambda$ and $\numcore$ cores.
\label{sec: FW CAM def}
\begin{definition}[\policyone]
\label{def:FW CAM}
Let $\kR{\vlambda}{n}$ be the optimal value of $k_i$ from the optimization problem \eqref{eq:online opt}, given a vector of arrival rates, $\vlambda$, and a time-average budget constraint of $n$ cores.\footnote{We assume the optimization~\eqref{eq:online opt} admits an optimal solution. The corner case where the infimum is not attained at any finite $k_i$ is not realistic but, for theoretical completeness, is handled in ~\ref{app:infinite k}.}
Pick $0.75<\beta<1$ and 
let $k_i:=\kR{\vlambda}{n-n^\beta}$ be the number of cores allocated by \policyone to each class-$i$ job.
We partition the $n$ cores into $m$ groups, where
class $i$'s group consists of $n_i$ cores.  
Here $n_i$ is defined to be proportional to the effective load contributed by class $i$.
That is,
\[\numcore_i := n\cdot \frac{r_i}{\sum_{j=1}^\numclass r_j}.\]
Whenever a class $i$ job arrives, if the total number of cores used by existing class $i$ jobs is no more than $\numcore_i-k_i$, the new job 
is allocated $k_i$ cores until it completes; otherwise, the new job queues. Whenever a class $i$ job completes, if there are queueing class $i$ jobs, the released $k_i$ cores are assigned to the class-$i$ job that has waited the longest; otherwise, if there are no queueing class $i$ jobs, the released cores idle even if there are queueing jobs from other classes.

Note that $n_i$ is chosen so that class $i$ jobs experience a system with subcritical effective load.
Statically reserving cores in this way is not necessary for implementation, but greatly simplifies our analysis.
\end{definition}

\subsection{The mean field optimality of \policyone}
\label{sec: FW CAM proof}
We now prove the mean field optimality of \policyone. 
Our proof has two major components.
First, we show in Lemma \ref{lemma: k converges} that, although the allocations of \policyone are slightly lower than the allocations of \budget (because \policyone has reserved some cores), the allocations used by \policyone converge to the allocations used by the Relaxed policy in the mean field limit.
Hence, the \emph{weighted service time} component of weighted mean response time under \policyone converges to that under the Relaxed policy.
Second, we show in the remainder of the section that the mean \emph{queueing time} under \policyone converges to 0.
Taken together, these results show that the mean response time under \policyone converges to the mean response time of the Relaxed policy.


We begin with the following lemma that compares the allocations of \policyone to the Relaxed policy in the mean field limit when both systems have $\numcore\inasym$ cores and a vector of arrival rates scaled up by $\systemindex$, which we denote as $\vlambda\inasym$.
This lemma will be used in Lemma~\ref{lemma: no queueing} and Theorem~\ref{thm:mean field optimal}.

\begin{lemma}
Let $k_i\inasym := k_i^*\left(\vlambda\inasym, \numcore\inasym - \left(\numcore\inasym\right)^\beta\right)$. 
For any class $i$,
    $\lim_{\systemindex\to\infty}k_i\inasym= \kR{\vlambda}{\numcore}$.  That is, the allocations under \policyone converge to the allocations of the Relaxed policy in the mean field limit.
    \label{lemma: k converges}
\end{lemma}

\begin{proof}[Proof Sketch]
Note that the optimization problem ~\eqref{eq:online opt} does not change when both the $\vlambda$ and $n$ are scaled by the same constant, $\systemindex$.  This gives 
\[\kR{\vlambda\inasym}{\numcore\inasym - {\left(\numcore\inasym\right)}^\beta} = \kR{\frac{\vlambda\inasym}{\systemindex}}{\frac{\numcore\inasym - {\left(\numcore\inasym\right)}^\beta}{\systemindex}}=\kR{\vlambda}{\numcore - \numcore^\beta \systemindex^{\beta-1}}.\]
Because $\beta<1$ we have
$$\lim_{\systemindex \to \infty} n-n^\beta\systemindex^{\beta-1} =n.$$
    The lemma then follows from the continuity and convexity of (\ref{eq:online opt}).  See \ref{app: k converge}.
\end{proof}

To show that the mean queueing time under \policyone goes to 0 in the mean field limit, it helps to first see how several of the critical quantities in the system will scale with $\systemindex$ under the \ourproblem policy.
In particular, we first define the effective load contributed by class $i$ jobs in the $\systemindex$th system
\begin{eqnarray}
\mbox{Effective load contributed by class } i = r_i\inasym= \frac{\mbox{expected cores used by class } i} {n\inasym }= \frac{\lambda_i\inasym}{\numcore\inasym} \cdot\frac{\E{X_i}k_i\inasym}{s_i(k_i\inasym)}. 
\end{eqnarray}
We can then examine the asymptotic behavior of the total effective load as follows:
\begin{equation*}
  \mbox{Expected total no. cores used }  = \sum_{i=1}^\numclass \frac{\lambda_i\inasym \E{X_i}k_i\inasym}{s_i(k_i\inasym)} = \numcore\inasym-\left(\numcore\inasym\right)^\beta.
\end{equation*}
\begin{equation}
   \mbox{Total effective load } =  \sum_{i=1}^\numclass r_i\inasym = \frac{\mbox{total expected cores used}}{n\inasym} = 1-\left(\numcore\inasym\right)^{\beta-1}.
    \label{eq:FW CAM total effective load}
\end{equation}
By definition,
$$n_i\inasym =  n\inasym \cdot \frac{r_i\inasym}{\sum_{j = 1}^m r_j\inasym}.$$

We now analyze the subsystem composed of the $n_i\inasym$ cores used to process class-$i$ jobs. Note that the cores are split into classes proportionately to the effective load, hence the effective load on this subsystem is the same as the total effective load. Mathematically, we have 
\begin{equation}
    \mbox{{Effective load on class-i subsystem }} = \frac{\lambda_i\inasym\left(\frac{\E{X_i}k_i\inasym}{s_i(k_i\inasym)}\right)}{n_i\inasym}=\frac{r_i\inasym n\inasym}{n_i\inasym} = \sum_{i=1}^\numclass r_i\inasym = 1-\left(\numcore\inasym\right)^{\beta-1}.
    \label{eq:eff_load_class_i}
\end{equation}
Hence, each job class experiences a subsystem that follows a Sub-Halfin-Whitt scaling regime that matches the overall system.

We now show that the mean queueing time under \policyone goes to 0 in the mean field limit. 
To do this, we start with Lemma \ref{label: reduction of system under MF-FW}, which shows that the system under \policyone can be decoupled into $\numclass$ independent FCFS (first-come-first-serve) GI/GI/$\numcore_i'$ systems where $\numcore_i'=\lfloor \frac{\numcore_i}{k_i} \rfloor$.
\begin{lemma}
Let ${\numcore_i'}\inasym=\left\lfloor  \frac{\numcore_i\inasym}{k_i\inasym} \right\rfloor$. Then
the mean response time of class-$i$ jobs under \policyone in the $\ourproblem\inasym$ system is the same as the mean response time in a single-class GI/GI/${\numcore_i'}\inasym$ system with average arrival rate $\lambda_i\inasym$, job sizes distributed as $\frac{X_i}{s_i(k_i\inasym)}$ and system load, ${r_i'}\inasym$, given by 
\begin{equation}
{r_i'}\inasym:=\frac{\lambda_i\inasym\left(\E{X_i}/s_i\left(k_i\inasym\right)\right)}{{n_i'}\inasym}.
\label{eq:ri_prime_d}
\end{equation}
\label{label: reduction of system under MF-FW}
\end{lemma}
\begin{proof}
    The proof is straightforward. Under \policyone,  each job class sees an independent system and each class-$i$ job occupies $k_i\inasym$ cores. Thus, the running time of each class-$i$ job follows the distribution $\frac{X_i}{s_i(k_i\inasym)}$. The $n_i\inasym$ cores reserved for class-$i$ jobs are divided into chunks of size $k_i\inasym$.  Here, ${n_i'}\inasym$ is the number of complete chunks of size $k_i\inasym$, hence the floor.  This creates a natural coupling between the system under \policyone and the GI/GI/$\numcore_i'$ system described above.  The load on the resulting system reflects the resources consumed by class-$i$ jobs (i.e., chunk-seconds used) divided by the number of chunks used.
\end{proof}

We now apply Lemma~\ref{lemma:goldberg}, adapted from \cite{li2025simple}, that bounds the mean number of jobs in queue in a GI/GI/$\numcore$ system.   Our statement below is simplified from the original statement in \cite{li2025simple} to aid in readability.
The full Lemma is provided in \ref{app:verification}.


\begin{lemma}[Adapted from Corollary 3 in \cite{li2025simple}]
For a GI/GI/$\numcore\inasym$ queue with arrival rate $\lambda\inasym$ and job size distribution $S$, let the system load, $\rho\inasym$, be defined as $\rho\inasym:= \frac{\lambda\inasym\E{S}}{n\inasym}< 1$ and let $\E{N_Q\inasym}$ be the mean number of jobs in queue.  If $\rho\inasym=1-\omega\left(\left(\numcore\inasym\right)^{-.25}\right),$ then 
$$\lim_{\systemindex \to \infty} \E{N_Q\inasym} = 0.$$
That is, if the system load of the GI/GI/$\numcore\inasym$ scales in a strongly Sub-Halfin-Whitt regime, the mean number of jobs in the queue goes to 0.
%
%
\label{lemma:goldberg}
\end{lemma}
We now use this result to finish our proof of the mean field optimality of \policyone.


\begin{lemma}[No queueing for each class]
\label{lemma: no queueing}
For any class $i$, the expected queueing time of class-$i$ jobs under \policyone goes to 0 when $\systemindex\to \infty$. 
\end{lemma}
\begin{proof}
By Lemma~\ref{label: reduction of system under MF-FW}, the mean queueing time of class-$i$ jobs is equal to the mean queueing time in the GI/GI/${n_i'}\inasym$ system described in Lemma~\ref{label: reduction of system under MF-FW}. Thus, we only need to show that the mean queueing time in the GI/GI/${n_i'}\inasym$ system goes to 0 when $\systemindex\to \infty$. 
It is straightforward to see that ${r_i'}\inasym$ from \eqref{eq:ri_prime_d} has the same asymptotic behavior as the effective load on the class-$i$ subsystem as described in \eqref{eq:eff_load_class_i}.
Hence, by Lemma \ref{lemma:goldberg}, the expected queueing time in the GI/GI/${n_i'}\inasym$ system goes to 0 as $d \to \infty$ (using Little's Law).
\end{proof}

Finally, we can prove our main theorem, which shows the mean field optimality of \policyone.

\begin{theorem}
\label{thm:mean field optimal}
    Our policy \policyone is mean field optimal.
\end{theorem}

\begin{proof}
    By Lemma~\ref{lemma: k converges}, as $\systemindex\to \infty$, the number of cores allocated to a class-$i$ job  converges to $\kR{\vlambda}{\numcore}$.   
    Hence
    the expected service time of class-$i$ jobs converges to $\frac{\E{X_i}}{s_i(k_i^*)}$. Moreover, by Lemma~\ref{lemma: no queueing}, the expected queueing time of class-$i$ jobs converges to 0. Thus, the expected response time of class-$i$ jobs converges to $\frac{\E{X_i}}{s_i(k_i^*)}$, and the weighted mean response time converges to
    \[\lim_{d\to \infty} \HC_{\ourproblem\inasym}^{\policyone} = \sum_{i=1}^\numclass \frac{\lambda_i}{\lambda} c_i\frac{\E{X_i}}{s_i(k_i^*)}, \]
    which is the weighted mean response time lower bound for \ourproblem established in \eqref{eq:budget is lower bound}.
\end{proof}

\section{A ``Better'' Mean-Field Optimal policy: \policytwo}
\label{sec: sec 6}

In Section~\ref{sec: FW Cap}, we prove the mean field optimality of \policyone. However, outside the mean field regime, \policyone may make frequent mistakes. For example, if the number of jobs in the system is very small, \policyone may leave many cores idle.
On the other hand, if there are too many jobs in the system, \policyone allocates many cores to a subset of jobs, while forcing other jobs to queue. 
Intuitively, outside the mean field regime, a better policy should take the {\em current number of jobs in the system} into consideration.
This raises the question of whether we can develop a policy that is more dynamic than \policyone while remaining mean field optimal.

For example, prior work \cite{berg2017towards} proposed the GREEDY policy\footnote{GREEDY is a generalization of the EQUI policy to workloads with multiple classes of jobs} that accounts for the number of jobs in the system by trying to maximize the sum of the speedups of all jobs in the system at every moment in time.
Specifically, GREEDY maximizes the current \emph{system efficiency} by repeatedly allocating cores to the jobs that receive the highest marginal improvement in speedup.
Using this heuristic, GREEDY tends to spread cores more evenly when there are many jobs in the system and skew allocations towards more parallelizable jobs when there are few jobs in the system.
Unfortunately, GREEDY is not mean field optimal in general (see \ref{app:greedy not mf}).

\noindent\textbf{Outline.}
In this section, we develop the \policytwo policy, a policy based on a Gittins index style of argument.
Our derivation of \policytwo shows that the \emph{correct} way to measure the marginal benefit a job receives from each core is the ratio of the reduction in a job's response time to the increase in the job's resource utilization.
Using this result, we prove the mean field optimality of \policytwo and demonstrate that prior approaches like GREEDY 
were overly simplistic.


The rest of this section will be devoted to deriving our heuristic policy, \policytwo, and then proving that it is asymptotically optimal in the mean field regime.
Following a Gittins-style argument (see e.g., \cite{gittins1979bandit,whittle1988restless,aalto2009gittins,Whittle_2005,nino2022multi,nino2001restless}), we begin by proposing a relaxation of \eqref{eq:opt:general} along two dimensions:
\begin{itemize}
\item[(i)] The first relaxation repeats our previous approach of using a time-average constraint on core usage rather than a strict constraint at every time, $t$ (see \eqref{eq:online opt}).
\item[(ii)] The second relaxation applies a discount factor relaxation that allows us to put more emphasis on the current allocation and not have to account for everything in the future.  
\end{itemize}

Relaxation (i) above is where we get the primary benefit, as it greatly simplifies the form of the optimization problem.   
Relaxation (ii) is necessary for technical reasons, and will be resolved at the end when we take a \emph{vanishing discount} limit (similar techniques are used in ~\cite{li2025improvinggeneralizedcmurule,larranaga2014index,Whittle_2005}).

We begin by setting up our Gittins approach in Sections \ref{sec:glam:relax} and \ref{sec:glam:solve}.
In Section \ref{sec:glam:gittins}, we derive the \policytwo policy.
This argument shows how to correctly measure a job's marginal benefit from each additional core, resulting in a policy that adapts allocations based on the number of jobs in the system.
Then, in Section \ref{sec: policy two proof}, we show that \policytwo is mean field optimal.

\subsection{Another Relaxed Optimization Problem}
\label{sec:glam:relax}
We first alter \eqref{eq:opt:general} by applying a discount factor of $e^{-\alpha t}$ to both the objective and constraint.
Here, we also re-write the objective function in terms of $N_i(t)$, the total number of class-$i$ jobs in the system at time $t$.
It is easy to see that with $\alpha=0$, this expression gives the total weighted response time across jobs over a time horizon, $\tau$.
For purposes of finding a more nuanced policy, we hold off on collapsing the objective and constraint into per-class expectations.

The constraint is relaxed in an analogous way, but requires tracking $\mathbb{N}_i(t)$, the {\em set} of class-$i$ jobs in the system at time $t$.
Note that the constraint on total discounted usage becomes $\frac{n}{\alpha}$ as a result of the discounting.
The resulting optimization problem is given as follows:

\begin{equation}
        \begin{aligned}
         \underset{k_{ij}(t)} {\text{minimize}} \quad & \lim_{\tau\to\infty} \int_0^\tau e^{-\alpha t}\sum_{i=1}^m N_{i}(t) c_i dt
        \\
         \text{subject to} \quad & \lim_{\tau\to\infty}\int_0^\tau e^{-\alpha t}\sum_{i = 1}^m \sum_{j \in \mathbb{N}_i(t)} k_{ij}(t) dt \leq \frac{\numcore}{\alpha} \\
         & k_{ij}(t) \in \{0\}\cup [1,\infty) \quad \forall t\geq 0.\\
        \end{aligned}
    \label{eq:opt:wham}
\end{equation}

%
%
\subsection{The Dual Problem}
\label{sec:glam:solve}
Given the optimization problem \eqref{eq:opt:wham}, we can use the method of Lagrange multipliers to set up an auxiliary optimization problem.
We begin by defining the Lagrangian as follows.
Let $\boldsymbol{k}_\pi(t)$ denote the vector of allocations to all jobs in the system at any time, $t$, under some scheduling policy, $\pi$.
Let $f(\kdot)$ denote the value of the objective function in \eqref{eq:opt:wham} for a given set of allocations.
Furthermore, let $g(\kdot)$ be the discounted time average core usage under these allocations minus $\frac{n}{\alpha}$.
That is, the constraint in \eqref{eq:opt:wham} is satisfied if $g(\kdot) \leq 0$.
We can then write the Lagrangian for this convex constrained optimization problem as follows:

$$\mathcal{L}(\kdot, \ell) = f(\kdot)+ \ell g(\kdot).$$

The Lagrange Multiplier Theorem says that, for $\ell \geq 0$, solving \eqref{eq:opt:wham} is equivalent to solving
$$\min_{\kdot}\max_{\ell} \mathcal{L}(\kdot, \ell).$$
Because \eqref{eq:opt:wham} is convex and strictly feasible (since $\rho < 1$ we can stabilize the system with fewer than $n$ cores), strong Lagrangian duality holds.  That is, we can define the Lagrangian dual function, 
$$h(\ell)=\min_{\kdot} \mathcal{L}(\kdot,\ell).$$
such that, to solve \eqref{eq:opt:wham}, it suffices to solve
\begin{equation}
\max_\ell h(\ell)= \max_\ell \min_{\kdot} \int_0^\infty e^{-\alpha t} \left[\sum_{i=1}^m N_i(t) c_i +\ell\cdot\sum_{j\in\mathbb{N}_i(t)} k_{ij}(t)\right] dt -\ell\cdot\frac{n}{\alpha}.
\label{eq:whittle}
\end{equation}

This expression for the dual problem can then be rewritten and simplified as follows:
\begin{equation}
\max_\ell h(\ell)= \max_\ell \min_{\kdot} \int_0^\infty e^{-\alpha t} \left[\sum_{i=1}^m  \sum_{j\in\mathbb{N}_i(t)} \left(c_i +\ell\cdot k_{ij}(t) \right)\right] dt -\ell\cdot\frac{n}{\alpha}.
\label{eq:whittle:simp}
\end{equation}

The only dependence between the terms inside the double summation in \eqref{eq:whittle:simp} is captured by the Lagrange multiplier, $\ell$.
Hence, for a given value of $\ell$, the problem of minimizing the inner expression in \eqref{eq:whittle:simp} (and thus yielding the value of $h(\ell)$) decouples into a series of separate optimization problems, one for each job that arrives to the system.
Specifically, for the $j$th job of class $i$, define an indicator $\one{ij}(t)$ which is one while the job is active in the system.
One must solve the following for each job:
\begin{equation}
    \min_{k_{ij}(t)} \int_{0}^{\infty} e^{-\alpha t} \one{ij}(t)\cdot\left(c_i + \ell\cdot k_{ij}(t)\right)dt.
    \label{eq:single}
\end{equation}
We refer to \eqref{eq:single} as the \emph{single job problem}.

We then note via complementary slackness that $\ell^*$, the value of $\ell$ that maximizes the Lagrangian dual function, is the value of $\ell$ where the discounted time average core usage across all jobs is exactly $\frac{n}{\alpha}$.
This suggests a two-step optimization process: 
First, analyze \eqref{eq:single} assuming arbitrary values of $\ell$.
Then, find the value,  $\ell^*$, such that the total core usage across jobs perfectly uses the available time average budget.



\subsection{\policytwo: A Gittins-like Policy}
\label{sec:glam:gittins}
Rather than solving \eqref{eq:whittle} via the two-step approach outlined above, we use \eqref{eq:whittle} to derive a heuristic for the original CAM problem based on the approach originally proposed by Gittins \cite{gittins1979bandit} and Whittle \cite{whittle1988restless}. 
We define the \emph{Gittins policy} for the relaxed problem as follows:
\begin{definition}[The Gittins policy for discounted problems]
    \label{def: whittle}
     At every moment in time, $t$, the Gittins policy calculates a ``market-clearing" service cost, $\widx$.  This $\widx$ is defined as the smallest value of $\ell$ such that, after solving the single job problems defined in \eqref{eq:single} for the jobs currently in the system,  the total number of cores used across all jobs is no more than $\numcore$.
     For each job, the Gittins policy sets $k_{ij}(t)$ to be the number of cores that optimizes \eqref{eq:single} given $\ell=\widx$.
\end{definition}

For this Gittins policy to be well-defined, there must exist a single $\widx$ satisfying the above conditions.  
This is generally shown by proving \emph{indexability}, as follows:

\begin{definition}[Indexability]
\label{def:indexability}
Given a discount factor $\alpha>0$,
define the {\em optimal action} for state $x$ in the single-armed bandit problem (parameterized by discount factor $\alpha$ and service penalty $\ell$) to be $k^*(x,\ell)$.
The discounted multi-armed bandit problem parameterized by the discount factor $\alpha$ is called {\em indexable} if for any state $x$, $k^*(x, \ell)$ is decreasing with $\ell$ and $\lim_{\ell\to 0} k^*(x,\ell)\to \infty$.  That is, when the service cost is lower, the optimal number of cores allocated is higher.
\end{definition}
We assume indexability holds in our setting.
Prior work \cite{nino2022multi} proves indexability under very general conditions that are close to our setting but require a finite state space for the underlying problem.
We defer formally extending these results to future work.
We will only use the assumption of indexability in deriving our Gittins policy.  The mean field optimality of our Gittins policy does not rely on indexability -- see Theorem~\ref{thm:mean field optimality of policytwo exp}.

We now derive \policytwo (Definition~\ref{def:policytwo}), which provides a closed form for the Gittins policy from the previous section. 
Here, we first derive the Gittins policy for the discounted problem, which does not have a simple closed form.
Fortunately, since we truly care about minimizing the undiscounted mean response time, we can take limits as $\alpha \to 0$.
This yields a closed form for the Gittins policy in the undiscounted case.
While this derivation of \policytwo yields a heuristic policy, \policytwo turns out to be asymptotically optimal (see Section \ref{sec: policy two proof}) and performs well in practice (see Section \ref{sec:simulation}).

Our central observation about \policytwo is that the correct way to measure the marginal benefit of allocating cores to jobs is to look at the \emph{ratio} of response time improvement to increase in resource usage for each additional core.
To see this, we define $f_i$ as follows:
\begin{definition}[$f_i$]
\label{def:fi}
For any class $i$ job of size $x$, define
\begin{equation}
    f_i(k):= \frac{\mbox{reduction in response time}}{\mbox{change in resource usage}}=\frac{-\frac{d}{dk} \left[\frac{x}{s_i(k)}\right]}{\frac{d}{dk}\left[\frac{kx}{s_i(k)}\right]}=\frac{s_i'(k)}{s_i(k)-ks_i'(k)}.
\end{equation}
Note the independence of the final expression from the job size, $x$.
Hence, $f_i$ is the same for all class-$i$ jobs.
Here the derivative of $s_i$ at 1 is defined as $s_i'(1):= \lim_{\epsilon\to 0^+} \frac{s_i(1+\epsilon)-1}{\epsilon}$. Further we define $f_i^{-1}$ to be the inverse function of $f_i$ (the inverse exists because $f_i$ is strictly decreasing, see Lemma~\ref{lemma:f decreasing}). 
\end{definition}



We will show that  \policytwo allocates cores to the jobs with the highest values of $f_i$.
Theorem \ref{thm:discount whittle} will show that, as $\alpha \to 0$, the value of $\widx$ falls into one of two cases.
Either 
$$\lim_{\alpha \to 0} \widx =\Theta\left(\frac{1}{\alpha}\right) \quad \mbox{or} \quad \lim_{\alpha \to 0} \widx =\Theta(1).$$
Hence, in Lemma \ref{lemma:single arm} below, it suffices to solve \eqref{eq:single} in these two limiting cases.

\begin{lemma}[Optimal policy for the single job problem]
\label{lemma:single arm}
The solution to the single job problem in the limit as $\alpha \to 0$ depends on the asymptotic behavior of $\widx$.

\noindent
If $\lim_{\alpha\to 0} \widx =\Theta(\frac{1}{\alpha})$, the optimal solution to the single job problem given a job of size $x$ is
\begin{equation*}
\begin{cases}
    \text{Choose action } k=1+o(1) & \text{if } \ \lim_{\alpha\to 0}\widx \alpha < \frac{c_i}{x} \\
    \text{Choose action } k=0 & \text{if } \ \lim_{\alpha\to 0}\widx \alpha > \frac{c_i}{x}.
\end{cases}
\end{equation*}
If $\lim_{\alpha\to 0} \widx = \Theta(1)$, the optimal solution to the single job problem given a job of size $x$ is
\begin{equation*}
    \begin{cases}
        \text{Choose action } k=1+o(1) & \text{if } \ \frac{\widx}{ c_i}\geq f_i(1) \\
        \text{Choose action } k=f_i^{-1}\left( \frac{\widx}{c_i}\right) + o(1) & \text{if } \ \frac{\widx}{c_i} < f_i(1).
    \end{cases}
\end{equation*}

\end{lemma}
\begin{proof}
By concavity of the speedup function, the optimal action for the single job problem should not change throughout the job's life. Otherwise, picking the time-average number of cores leads to a better policy (see Lemma~\ref{lemma: app single arm} in \ref{app: proof for SAB} for a detailed proof). Hence, we only need to focus on the total cost of each action. 

    The total discounted cost for action $k=0$ is
    \begin{equation}
        \int_0^\infty e^{-\alpha s} c_i ds = \frac{c_i}{\alpha}.
        \label{eq:cost k=0}
    \end{equation}
    For any action $k\geq 1$, the total discounted cost is 
    \begin{align}
        \int_0^{\frac{x}{s_i(k)}} e^{-\alpha s} (c_i + \ell \cdot k) ds \nonumber = \frac{1}{\alpha} \left(c_i + \ell k \right) \cdot \left(1-e^{-\alpha \frac{x}{s_i(k)}}\right) \nonumber&= \frac{1}{\alpha} \left(c_i + \ell k \right)\left( \alpha \frac{x}{s_i(k)} + o(\alpha) \right) \nonumber \\
        &= \frac{\ell k x}{s_i(k)} + \frac{c_i x}{s_i(k)} + o(1)+o(\ell)k.
        \label{eq:eq12}
    \end{align}
We now derive the optimal action for the single job problem for the two cases in Lemma \ref{lemma:single arm}.

\noindent\textbf{Case 1: $\ell=\Theta(\frac{1}{\alpha})$.}
In this case, the comparison between \eqref{eq:cost k=0} and \eqref{eq:eq12} depends on the sign of $\ell\alpha-\frac{c_i}{x}$.

If $\lim_{\alpha\to 0}\ell\alpha > \frac{c_i}{x}$, we see that \eqref{eq:cost k=0} (cost when $k=0$) is smaller than \eqref{eq:eq12} for any $k\geq 1$.
Specifically, by sublinearity ($s_i(k)\le k$), for any $k\ge 1$,
\begin{equation*}
    \frac{\ell k x}{s_i(k)} + \frac{c_i x}{s_i(k)} + o(1) + o(\ell)\,k
    \;\ge\; \ell x + O(1) + o(\ell).
\end{equation*}
The condition $\lim_{\alpha\to 0}\ell\alpha>\frac{c_i}{x}$ gives $\ell x - \frac{c_i}{\alpha} = \Theta(\frac{1}{\alpha})>0$. Hence, for any $k\geq1$, \eqref{eq:eq12} is greater than \eqref{eq:cost k=0} for sufficiently small $\alpha$.
This implies that $k=0$ is the optimal action as $\alpha \to 0$ in this case.

If $\lim_{\alpha\to 0}\ell\alpha < \frac{c_i}{x}$, we see that $k=1+o(1)$ is optimal as $\alpha\to 0$: By \eqref{eq:eq12}, the cost of any $k\ge 1$ is $\frac{k}{s_i(k)}\,\ell x + O(1) + o(\ell)k$. At $k=1$ this is $\ell x + O(1) + o(\ell)$.
The assumption $\lim_{\alpha\to 0}\ell \alpha < c_i/x$ now gives $\ell x - \frac{c_i}{\alpha} = \Theta(\frac{-1}{\alpha})<0$.
Hence, the cost of using $k=1$ is less than the cost of using $k=0$ for sufficiently small $\alpha$.
For any $k> 1$, strict sub-linearity of the speedup function gives $\frac{k}{s_i(k)}> 1$, so the cost of using $k>1$ exceeds the cost of using $k=1$ by $\Theta(\frac{1}{\alpha})$.
This cost difference dominates any differences in the $O(1)$ and $o(\ell)$ terms. Hence the optimal $k$ equals $1+o(1)$.

Combining the two sub-cases, the optimal action is
\begin{equation*}
\begin{cases}
    \text{Choose action } k=1+o(1) & \text{if } \ \lim_{\alpha\to 0}\ell\alpha < \frac{c_i}{x} \\
    \text{Choose action } k=0 & \text{if } \ \lim_{\alpha\to 0}\ell\alpha > \frac{c_i}{x}.
\end{cases}
\end{equation*}

\noindent\textbf{Case 2: $\ell=\Theta(1)$.} Since $\ell=\Theta(1)$, the term $o(\ell)\,k$ in \eqref{eq:eq12} is $o(1)$ (because $k\leq \numcore$). Hence the cost of any $k\ge 1$ is
\begin{equation*}
    \frac{(\ell k + c_i)\,x}{s_i(k)} + o(1) \;=\; \Theta(1),
\end{equation*}
which is less than the cost $\frac{c_i}{\alpha}=\Theta(\frac{1}{\alpha})$ of $k=0$. So the optimal action is some $k\ge 1$.

We first minimize the leading part $\frac{(\ell k + c_i)x}{s_i(k)}$, then deal with the $o(1)$ term. Differentiating,
\begin{equation*}
    \frac{d}{dk}\!\left[\frac{(\ell k+c_i)x}{s_i(k)}\right]
    \;=\; \frac{c_i x\,(s_i(k)-k s_i'(k))}{s_i(k)^2}\left(\frac{\ell}{c_i}-f_i(k)\right).
\end{equation*}
By strict concavity of $s_i$, $s_i(k)-ks_i'(k)>0$; also $f_i$ is strictly decreasing. Therefore:
\begin{itemize}
    \item If $\frac{\ell}{c_i}\geq f_i(1)$, then $f_i(k)\le f_i(1)\le \frac{\ell}{c_i}$, so the derivative is $\ge 0$ throughout, and the leading part is minimized at $k=1$.
    \item If $\frac{\ell}{c_i}< f_i(1)$, then there is a unique $k>1$ where the derivative equals 0, namely $k=f_i^{-1}(\frac{\ell}{c_i})$; the leading part is minimized at this value of $k$.
\end{itemize}
Let $\hat{k}$ be the value of $k$ that minimizes this leading term.

We now show that the $o(1)$ term shifts our optimal choice $k$ by at most $o(1)$.
Note that by the derivative formula above, the leading part is strictly decreasing before $\hat{k}$ and strictly increasing after. Hence, for any fixed $\epsilon>0$, the leading part exceeds its minimum by at least some positive constant $c(\epsilon)>0$ whenever $k$ is at distance at least $\epsilon$ from $\hat{k}$. Reduction in the $o(1)$ term in \eqref{eq:eq12} cannot compensate for this $c(\epsilon)=\Theta(1)$ increase as $\alpha\to 0$, so the optimal $k$ is within $\epsilon$ of $\hat{k}$. Since $\epsilon$ is arbitrary, the optimal $k$ differs from $\hat{k}$ by at most $o(1)$, i.e.,
\begin{equation*}
\begin{cases}
    \text{Choose action } k=1+o(1) & \text{if } \ \frac{\ell}{c_i}\geq f_i(1) \\
    \text{Choose action } k=f_i^{-1}\!\left(\frac{\ell}{c_i}\right)+o(1) & \text{if } \ \frac{\ell}{c_i}< f_i(1).
\end{cases}
\end{equation*}
\end{proof}

We are now ready to derive the Gittins policy for scheduling a stream of jobs.
Note that a key step in deriving the Gittins policy (Definition~\ref{def: whittle}) is to pick the ``market-clearing'' service cost $\widx$. 
In Theorem \ref{thm:discount whittle} below, we show that the scaling behavior of $\widx$ depends on the current number of jobs in the system.
This allows us to derive a closed form for the Gittins policy as $\alpha \to 0$.

\begin{theorem}
\label{thm:discount whittle}
    Assuming indexability, the Gittins policy in the limit as $\alpha\to 0$ falls into one of two cases at every moment in time:
    \begin{enumerate}[leftmargin=*]
        \item If the number of jobs is larger than the number of cores, then $\widx$ scales as $\Theta(1/\alpha).$ The Gittins policy assigns each class-$i$ job an index of $\left(\frac{c_i}{\text{remaining job size}}+o(1)\right)$.  The policy gives $1+o(1)$ cores to the job with the highest index until insufficient cores remain for the next job.
        
        \item Otherwise, $\widx$ scales as $\Theta(1)$ and is chosen such that, if every class-$i$ job gets $g_i(\widx)$ cores, the total number of cores used is no more than $\numcore$. Here $g_i(\ell)$ is defined as 
    \begin{equation*}
        g_i(\ell):=
    \begin{cases}
        1+o(1) & \text{if } \ \frac{\ell}{c_i}\geq f_i(1), \\
        f_i^{-1}\left( \frac{\ell}{c_i} \right)+o(1)& \text{otherwise,}
    \end{cases}
\end{equation*}
and where $f_i$ is defined in Definition~\ref{def:fi}. The Gittins policy allocates every class-$i$ job  $g_i(\widx)$ cores.
    \end{enumerate}
\end{theorem}
\begin{proof}
    By Definition~\ref{def: whittle}, the Gittins policy is determined by selecting the value of $\ell$ such that the total used cores is equal to $\numcore$. 
    Now we discuss two cases depending on the number of jobs.

\noindent\textbf{Case 1: number of jobs $\mathbf{\geq \numcore}$.} 
    In this case, as $\alpha\to 0$, if $\widx \in \Theta(1)$, by Lemma~\ref{lemma:single arm} we know that each job wants at least 1 core, which will exceed our total limit $\numcore$. Thus, we should make the service cost higher to decrease the total core usage.
    By Lemma~\ref{lemma:single arm}, whether a job with weight $c$ and remaining inherent work $x$ wants 1 or 0 cores depends on whether $\widx \alpha>\frac{c}{ x}+o(1)$. 
    Hence, $\widx$ must scale as $\Theta\left(\frac{1}{\alpha}\right)$ in order to fulfill the market clearing condition as $\alpha \to 0$.

    By Lemma~\ref{lemma:single arm},
    given a value of $\widx$, a job is allocated $1+o(1)$ cores if and only if $\widx \alpha < \frac{c_i}{x} + o(1)$. Hence, 
    the Gittins policy allocates $1+o(1)$ cores to each job according to the order of $\frac{c}{x}+o(1)$.


\noindent\textbf{Case 2: number of jobs $\mathbf{<\numcore}$.}
    In this case, because one or more jobs will run on more than 1 core, the market-clearing $\widx$ must scale as $\widx \in \Theta(1)$.
    Specifically, by Lemma~\ref{lemma:single arm}, $\widx \in \Theta(1)$ implies that every class-$i$ job gets $1+o(1)$ cores if $\frac{\widx}{c_i}\geq f_i(1)$, and gets $f_i^{-1}(\frac{\widx}{c_i}) + o(1)$ cores otherwise. Hence, under the market-clearing $\widx$, each job gets $g_i(\widx)$ cores and the market-clearing $\widx$ is chosen accordingly.

In both of these cases, a unique market-clearing service cost $\widx$ exists because of indexability.
\end{proof}

Finally, we define the \policytwo policy, which considers the Gittins policy from Theorem~\ref{thm:discount whittle} in the limit as $\alpha\to 0$, causing all $o(1)$ and $o(\alpha)$ terms to vanish:

\begin{definition}[\policytwo]
\label{def:policytwo}
    At every moment of time:
    \begin{itemize}[leftmargin=*]
        \item If the number of jobs in the system is at least the number of cores, assign each class-$i$ job an index of $\frac{c_i}{\text{remaining job size}}$.
        Allocate 1 core to each of the $\numcore$ jobs with the largest index.
        \item Otherwise, choose the smallest $\ell$ such that if every class-$i$ job gets $g_i(\ell)$ cores, the total number of cores used is no more than $\numcore$, where $g_i(\ell)$ is defined as
        \begin{equation*}
        g_i(\ell):=
    \begin{cases}
        1 & \text{if } \frac{\ell}{c_i}\geq f_i(1), \\
        f_i^{-1}\left( \frac{\ell}{c_i}\right)& \text{otherwise,}
    \end{cases}
\end{equation*}
where $f_i$ and $f_i^{-1}$ are defined in Definition~\ref{def:fi}. The chosen value of $\ell$ represents $\widx$, and every class-$i$ job should be allocated $g_i(\widx)$ cores.
    \end{itemize}
\end{definition}
There are two important intuitions to note about the \policytwo policy.
First, if we take the classical interpretation of $\ell$ as a ``service charge'' assessed to a job for every core it uses, the form of $g_i$ has a nice interpretation.
Here, $\frac{\ell}{c_i}$ represents the marginal cost of using an additional core divided by the marginal benefit of reducing the response time of a class-$i$ job by one second.
The $f_i$ function similarly looks at the ratio of the marginal change in response time to the marginal change in resource usage.
By equating these two ratios, we get
$$\parbox{6em}{\centering change in response time} \quad\times\quad \parbox{10em}{\centering marginal benefit of response time reduction} \quad\scalebox{1.5}{=}\quad \parbox{6em}{\centering change in resource usage} \quad\times\quad \parbox{7em}{\centering marginal cost of resource usage}.$$
This matches the classical economic tenet of optimizing profit by setting marginal cost equal to marginal benefit.
The \policytwo policy then sets a price, $\widx$, such that when each job optimizes its own profit, the resource constraint becomes tight.

Second, one practical way to think about \policytwo's allocations is to think about the policy repeatedly allocating one core at a time (or some small fraction of a core) until all cores are allocated.
To do this, \policytwo roughly awards each core to the job with the highest value of $f_i(k)\cdot c_i$ under the current allocations.
By contrast, the GREEDY policy has been described as awarding cores to the jobs with the highest value of $s_i'(k)$.
This has a similar flavor to \policytwo, but simply uses an overly simplistic heuristic to gauge a job's marginal improvement from each additional core.

\subsection{Mean field optimality of \policytwo}
\label{sec: policy two proof}

While we have both a theoretical justification and a strong intuition for \policytwo, at this point, it is still just another heuristic policy.
Hence, in this section, we prove the mean field optimality of \policytwo (Definition~\ref{def:policytwo}) subject to some mild conditions. 
As stated in Section \ref{sec:setting}, here we assume that job sizes are phase-type distributions.
This is a commonly used assumption because phase-type distributions can approximate any distribution arbitrarily closely (Theorem~4.2 in \cite{asmussen2003applied}). Without such an assumption, capturing the dynamics of the system in mean field would require using a {\em measure valued fluid approximation} (e.g., \cite{kaspi2011law}), which is beyond the scope of this paper.

For ease of presentation, we present our proof for the case when job sizes are exponential. Specifically, class-$i$ jobs have size $X_i \sim \Exp(\mu_i)$. The proof for phase type distributions follows a very similar roadmap, and is deferred to \ref{app:ph type}.


In the mean field regime, the limiting dynamics of the queueing system can be captured by an ordinary differential equation (ODE)~\cite{kurtz1970solutions,gast2017refined}.
Specifically, let ${N_i(t)}\inasym$ be the number of class-$i$ jobs in the $\systemindex$th system at time $t$.  
Let $A_i^{(d)}(t)$ and $D_i^{(d)}(t)$ denote the numbers of class-$i$ arrivals and completions by time $t$, respectively, and define the fluid-scaled process $z_i^{(d)}(t):=N_i^{(d)}(t)/d$. The population balance is
\begin{equation*}
z_i^{(d)}(t)=z_i^{(d)}(0)+\frac{A_i^{(d)}(t)}{d}-\frac{D_i^{(d)}(t)}{d}.
\end{equation*}
By definition of the policy \policytwo, 
every class-$i$ job then receives $g_i(\ell^{(d)}(u))$ cores, where $g_i$ is defined in Definition~\ref{def:policytwo}, and $\ell^{(d)}(u)$ is picked such that $\sum_i z_i^{(d)}(u)g_i(\ell^{(d)}(u))=\numcore$.{\footnote{In the mean field limit, the number of jobs in the system is never larger than the number of cores. }}  Because class-$i$ job sizes are exponential with rate $\mu_i$, the completion process can be represented as: 
\begin{equation*}
D_i^{(d)}(t)=Y_i\!\left(d\!\int_0^t z_i^{(d)}(u)\mu_i
s_i\!\left(g_i(\ell^{(d)}(u))\right)du\right),
\end{equation*}
where $Y_i$ is a unit-rate Poisson process. Applying the functional laws of large numbers for the renewal and Poisson processes to the scaled population balance yields the associated fluid equation
\begin{equation*}
z_i(t)=z_i(0)+\lambda_i t-\int_0^t z_i(u)\mu_i
s_i\!\left(g_i(\ell(u))\right)du.
\end{equation*}
{Consequently, by taking derivative on both sides of the above equation,} the \policytwo system in the mean field limit can be described by the following ODE (where $\dot z_i$ denotes $\frac{d z_i(t)}{dt}$ and $z_i$ is used as shorthand for $z_i(t)$):
\begin{align}
    &\dot z_i = \lambda_i - z_i\cdot \mu_i \cdot s_i(g_i(\ell)), \label{eq:ODE exp}
\end{align}
where $\ell$ is determined by $\sum_{i=1}^\numclass z_i\cdot g_i(\ell) = \numcore$ and $g_i$ is defined in Definition~\ref{def:policytwo}.

We can now state our main theorem. The mean field optimality of \policytwo relies on the assumption that ODE \eqref{eq:ODE exp} has a global attractor. Similar assumptions of the existence of a global attractor are also made in prior works (e.g., \cite{weber1990,Verloop_2016}, see \cite{Verloop_2016} for a detailed discussion).\footnote{It is common practice to assume that ODE's such as \eqref{eq:ODE exp} have a global attractor since the proof of this property is highly non-trivial even when job sizes are exponential.  We leave exploration of this assumption for future work.}
\begin{theorem}[mean field optimality of \policytwo]
    \label{thm:mean field optimality of policytwo exp}
    Assuming ODE~\eqref{eq:ODE exp} has a global attractor, the \policytwo policy is mean field optimal. 
\end{theorem}
\begin{proof}
Suppose ODE~\eqref{eq:ODE exp} has a global attractor, then the global attractor must be a stationary point. We next prove that there exists a unique stationary point for ODE~\eqref{eq:ODE exp}, and the weighted mean response time for the stationary point is equal to the optimal weighted mean response time of the Relaxed problem (see \eqref{eq:online opt}). Note that if this statement holds, we have that \policytwo is mean field optimal because its weighted mean response time converges to the lower bound given by the Relaxed problem.

    For the stationary point $z_i$, since $\dot{z}_i=0$, we have that $z_i \mu_i s_i(g_i(\ell))= \lambda_i.$
    Thus we have that 
    \begin{equation}
        z_i = \frac{\lambda_i}{s_i(g_i(\ell))}\E{X_i}.
        \label{eq:14 exp}
    \end{equation}
    Now we can find the stationary point by solving the equation $\sum_{i=1}^\numclass z_i\cdot g_i(\ell)=\numcore$,
    which is equivalently
    \begin{equation}
        \sum_{i=1}^\numclass \frac{\lambda_i g_i(\ell)}{s_i(g_i(\ell))}\E{X_i} = n.
        \label{eq:eq15 exp}
    \end{equation}

    We now argue that there exists a unique $\ell^*<\max_i c_if_i(1)$ solving \eqref{eq:eq15 exp}. First, for any $\ell\geq \max_i c_if_i(1)$, so we have $g_i(\ell)=1$ for all $i$, and the left hand side of \eqref{eq:eq15 exp} becomes $\sum_{i}\lambda_i \E{X_i}<n$.  Next, note that $g_i(\ell)$ is strictly decreasing for any $\ell< c_if_i(1)$, hence the left hand side of \eqref{eq:eq15 exp} is strictly decreasing with $\ell<\max_i c_if_i(1)$. Finally, note that as $\ell$ goes to 0, $g_i(\ell)=f^{-1}_i(\ell/c_i)$ goes to $\infty$.  To see this, note that the numerator of $f_i$ is $o(1)$ by strict sub-linearity and the denominator is $\Omega(1)$ by concavity (specifically the derivative of the denominator is increasing).
    Hence, there exists a unique $\ell^*$ solving \eqref{eq:eq15 exp} in the interval $(0,\max_i c_i f_i(1)]$.
    Therefore, the unique stationary point $z_i$ can be found by plugging $\ell^*$ into \eqref{eq:14 exp}.
    
    Since the dynamics of the \policytwo system in the mean field limit are captured by an ODE~\eqref{eq:ODE exp} whose global attractor is $z_i$, we have that in the limit, any class-$i$ job is assigned $g_i(\ell^*)$ cores. Thus
    the mean response time of a class-$i$ job is $\frac{\E{X_i}}{s_i(g_i(\ell^*))}$, and the weighted mean response time in the limit is
    \[\lim_{\systemindex\to\infty} \HC_{\ourproblem\inasym}^{\policytwo} = \frac{1}{\lambda}\sum_{i=1}^\numclass \frac{c_i\lambda_i\E{X_i}}{s_i(g_i(\ell^*))}.\]
    One can use Lagrange multipliers to solve the optimization problem~\eqref{eq:online opt} and find that the optimal weighted mean response time for the Relaxed problem is $ \frac{1}{\lambda}\sum_{i=1}^\numclass \frac{c_i\lambda_i\E{X_i}}{s_i(g_i(\ell^*))}$ (Lemma~\ref{lemma:appendix D}). Thus we have that \policytwo is mean field optimal.
\end{proof}

\section{Simulation and Evaluation}
\label{sec:simulation}

We evaluate the performance of our proposed policies, \policyone and \policytwo, through numerical simulations. Our evaluation focuses on two aspects: (1) verifying the mean field optimality of our policies and comparing their convergence rates, and (2) demonstrating the superiority of our policies (particularly \policytwo) over standard heuristics commonly used in the literature. We ran simulations over a wide range of settings, varying the speedup functions, the system load $\rho$, and the inter-arrival time and job size distributions. Across all settings the message is the same:
\begin{itemize}[leftmargin=*]
    \item Both \policyone and \policytwo converge to the theoretical lower bound given by the Relaxed optimization~\eqref{eq:online opt}, while standard heuristics from the literature do not.
    \item \policytwo converges significantly faster than \policyone and beats the heuristics even when the number of cores is small.
\end{itemize}
For brevity, the main text reports one \emph{representative} experiment in Figure~\ref{fig:simulation_results}; additional experiments under other inter-arrival time and job size distributions are deferred to \ref{app:sim}.

\paragraph{Convergence rates}
Figure~\ref{fig:simulation_results}(a) plots the mean response time across jobs as a function of $n$ for each of our policies. 
Here, we scale $n \to \infty$ in a mean field scaling regime with a fixed load of $\rho=0.6$.
While both of our policies are mean field optimal, \policytwo converges significantly faster than \policyone, and achieves near-optimal performance even when the number of cores is small.

\paragraph{Comparison to heuristics}
We compare \policytwo against existing heuristics in the literature. None of these heuristics is mean field optimal, for one of two reasons. (1) The size-based heuristics HELL~\cite{lin2018model} and (A-)heSRPT~\cite{berg2021hesrpt} prioritize jobs by their remaining size, and are therefore not mean field optimal by Proposition~\ref{prop:app no job size} which shows that no mean field optimal policy can depend on exact job sizes in the limit. (2) The size-blind heuristics EQUI~\cite{edmonds1999scheduling}, KNEE~\cite{lin2018model}, and GREEDY~\cite{berg2017towards} can instead be shown to be strictly suboptimal in mean field via the ODE method of Section~\ref{sec: policy two proof}: \ref{app:greedy not mf} carries out this argument explicitly for GREEDY; similar arguments apply to EQUI and KNEE.

We implement two representative heuristics in our simulations:

    \noindent\textbf{EQUI}~\cite{edmonds1999scheduling}: A fairness-based policy that, at every moment, divides the total available cores equally among all active jobs in the system.

    \noindent\textbf{GREEDY}~\cite{berg2017towards}: An efficiency-based policy that, at every moment, maximizes the total rate the system completes inherent work. GREEDY has recently been deployed in real systems (e.g., \cite{qiao2021pollux,jayaram2023sia}).

Figure~\ref{fig:simulation_results}(b) compares \policytwo against EQUI and GREEDY. The results demonstrate that while \policytwo converges to the theoretical optimum, the heuristic policies do not. Furthermore, \policytwo consistently outperforms these heuristics for any number of cores, demonstrating its effectiveness even outside of the mean field limit.


\begin{figure}[t]
    \centering
    \begin{minipage}{0.4\textwidth}
        \centering
        \includegraphics[width=\textwidth]{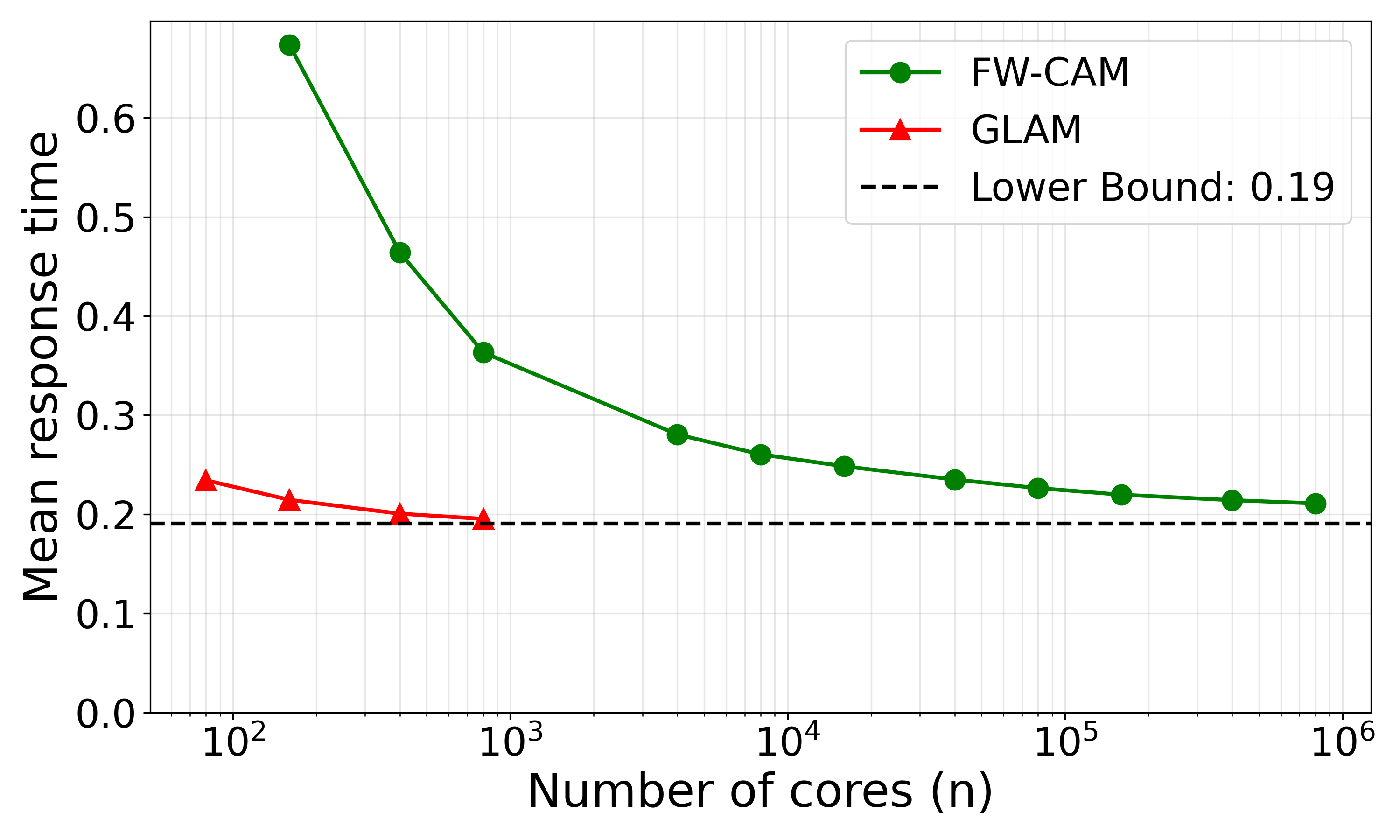}
        \\[2pt]
        (a) \policyone vs. \policytwo
    \end{minipage}
    \qquad
    \begin{minipage}{0.4\textwidth}
        \centering
        \includegraphics[width=\textwidth]{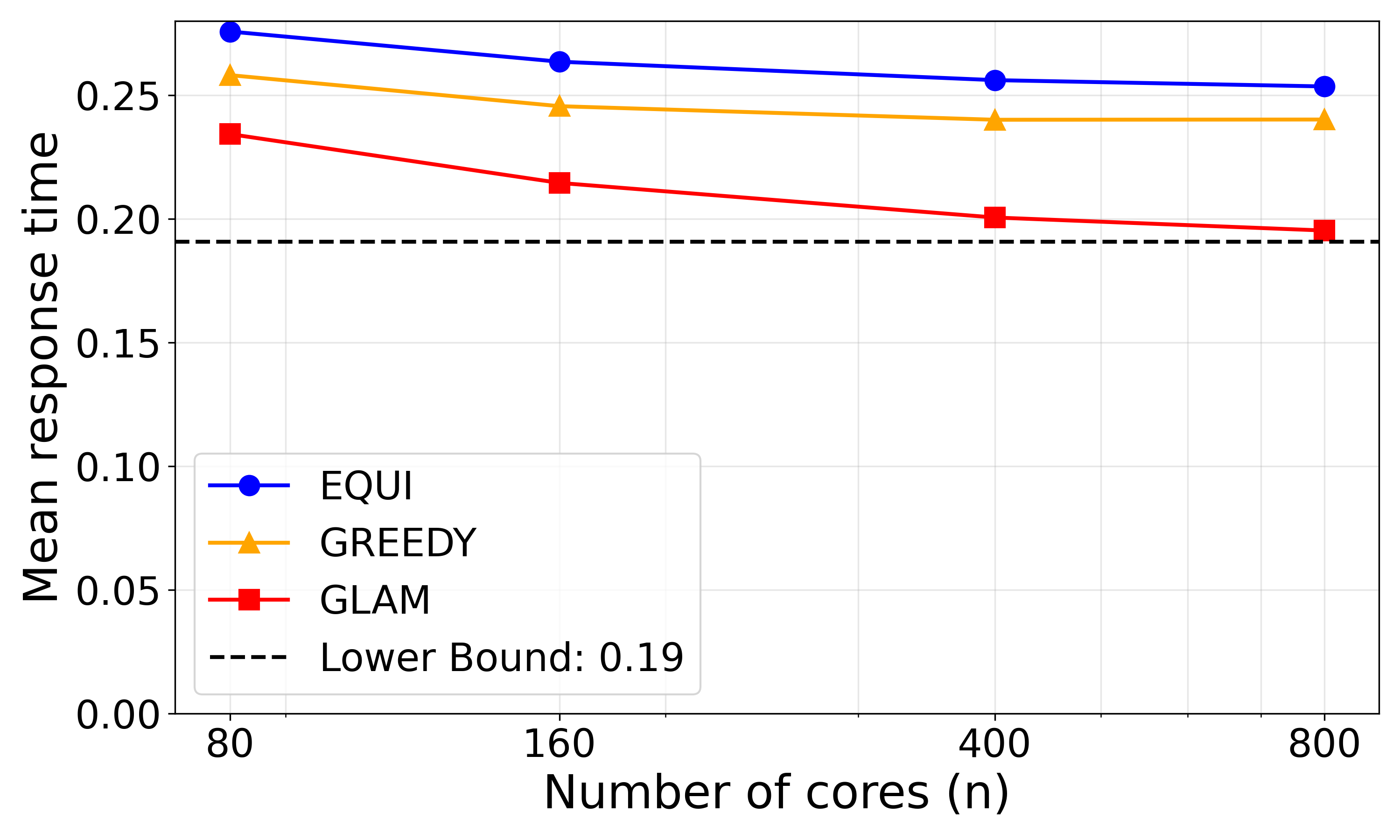}
        \\[2pt]
        (b) \policytwo vs. EQUI and GREEDY
    \end{minipage}
    \caption{Performance evaluation of core allocation policies under a \emph{representative} setting: two classes with sub-linear speedup functions, Poisson arrivals, and Exponential job sizes, with system load fixed at $\mathbf{\rho=0.6}$ as $n$ scales. The same qualitative behavior is observed across all other settings we tested; see ~\ref{app:sim} for additional inter-arrival time and job size distributions and for full workload details.}
    \label{fig:simulation_results}
\end{figure}



\section{Conclusion}
\label{sec:conclusion}

Prior work on the \ourproblem problem \cite{berg2017towards,berg2021hesrpt} conjectured that optimizing core allocation involves both favoring short jobs and maximizing the total system efficiency.  
However, general optimality results have remained elusive. 
By deriving the \policyone and \policytwo policies, we see that neither of the above intuitions holds in the mean field regime.
In particular, our derivation of \policyone shows us that any policy that favors short jobs is suboptimal.  
Likewise, our derivation of the \policytwo policy illustrates that maximizing system efficiency, as done by the GREEDY policy, is suboptimal.

This paper derives the first asymptotic optimality results for the \ourproblem problem with multiple concave speedup functions.
Our policies \policyone and \policytwo are asymptotically optimal under very general distributional assumptions on job sizes and inter-arrival times.


\section*{Acknowledgments}
This work was supported by the National Science Foundation under grants
CIF-2403194, III-2322973, CMMI-2307008, CIF-2403195, and III-2322974,
and by a Google Junior Faculty Award.

\clearpage
\bibliographystyle{elsarticle-num}
\bibliography{bib}

@misc{li2024rentgpusbudget,
      title={How to Rent {GPUs} on a Budget}, 
      author={Zhouzi Li and Benjamin Berg and Arpan Mukhopadhyay and Mor Harchol-Balter},
      year={2024},
      eprint={2406.15560},
      archivePrefix={arXiv},
}

@phdthesis{baccelli1985simple,
  title={Simple computable bounds for the fork-join queue},
  author={Baccelli, Fran{\c{c}}ois and Makowski, A},
  year={1985},
  school={INRIA}
}

@article{nelson2002approximate,
  title={Approximate analysis of fork/join synchronization in parallel queues},
  author={Nelson, Randolf D. and Tantawi, Asser N.},
  journal={IEEE Transactions on Computers},
  volume={37},
  number={6},
  pages={739--743},
  year={1988},
  publisher={IEEE}
}

@article{wang2019delay,
  title={Delay asymptotics and bounds for multi-task parallel jobs},
  author={Wang, Weina and Harchol-Balter, Mor and Jiang, Haotian and Scheller-Wolf, Alan and Srikant, Rayadurgam},
  journal={ACM SIGMETRICS Performance Evaluation Review},
  volume={46},
  number={3},
  pages={2--7},
  year={2019},
  publisher={ACM New York, NY, USA}
}

@article{kim1989analysis,
  title={Analysis of the fork-join queue},
  author={Kim, Cheeha and Agrawala, Ashok K},
  journal={IEEE Transactions on Computers},
  volume={38},
  number={2},
  pages={250--255},
  year={1989},
  publisher={IEEE}
}

@article{aalto2023minimizing,
  title={Minimizing the mean slowdown in the {M/G/1} queue},
  author={Aalto, Samuli and Scully, Ziv},
  journal={Queueing Systems},
  volume={104},
  number={3},
  pages={187--210},
  year={2023},
  publisher={Springer}
}

@article{li2025simple,
  title={Simple and explicit bounds for multiserver queues with 1/1- $\rho$ scaling},
  author={Li, Yuan and Goldberg, David A},
  journal={Mathematics of Operations Research},
  volume={50},
  number={2},
  pages={813--837},
  year={2025},
  publisher={INFORMS}
}

@article{gittins1979bandit,
  title={Bandit processes and dynamic allocation indices},
  author={Gittins, John C},
  journal={Journal of the Royal Statistical Society Series B: Statistical Methodology},
  volume={41},
  number={2},
  pages={148--164},
  year={1979},
  publisher={Oxford University Press}
}

@article{nino2022multi,
  title={Multi-gear bandits, partial conservation laws, and indexability},
  author={Ni{\~n}o-Mora, Jos{\'e}},
  journal={Mathematics},
  volume={10},
  number={14},
  pages={2497},
  year={2022},
  publisher={MDPI}
}

@article{Verloop_2016,
   title={Asymptotically optimal priority policies for indexable and nonindexable restless bandits},
   volume={26},
   ISSN={1050-5164},
   number={4},
   journal={The Annals of Applied Probability},
   publisher={Institute of Mathematical Statistics},
   author={Verloop, I. M.},
   year={2016},
   month=aug }

@inproceedings{Arpan2024,
author = {Ghanbarian, Samira and Mukhopadhyay, Arpan and Mazumdar, Ravi R. and Guillemin, Fabrice M.},
title = {On Optimal Server Allocation for Moldable Jobs with Concave Speed-Up},
year = {2024},
isbn = {9798400705212},
publisher = {Association for Computing Machinery},
address = {New York, NY, USA},
booktitle = {Proceedings of the Twenty-Fifth International Symposium on Theory, Algorithmic Foundations, and Protocol Design for Mobile Networks and Mobile Computing},
pages = {191–200},
numpages = {10},
location = {Athens, Greece},
series = {MobiHoc '24}
}

@article{berg2024asymptotically,
  title={Asymptotically optimal scheduling of multiple parallelizable job classes},
  author={Berg, Benjamin and Moseley, Benjamin and Wang, Weina and Harchol-Balter, Mor},
  journal={arXiv preprint arXiv:2404.00346},
  year={2024}
}

@article{halfin1981heavy,
  title={Heavy-traffic limits for queues with many exponential servers},
  author={Halfin, Shlomo and Whitt, Ward},
  journal={Operations Research},
  volume={29},
  number={3},
  pages={567--588},
  year={1981},
  publisher={Informs}
}

@article{kingman1962inequalities,
  title={Some inequalities for the queue {GI/G/1}},
  author={Kingman, J. F. C.},
  journal={Biometrika},
  volume={49},
  number={3/4},
  pages={315--324},
  year={1962},
  publisher={Oxford University Press}
}

@article{weber1990,
 ISSN = {00219002},
 author = {Richard R. Weber and Gideon Weiss},
 journal = {Journal of Applied Probability},
 number = {3},
 pages = {637--648},
 publisher = {Applied Probability Trust},
 title = {On an Index Policy for Restless Bandits},
 urldate = {2025-12-21},
 volume = {27},
 year = {1990}
}

@article{whittle1988restless,
  title={Restless bandits: Activity allocation in a changing world},
  author={Whittle, Peter},
  journal={Journal of Applied Probability},
  volume={25},
  number={A},
  pages={287--298},
  year={1988},
  publisher={Cambridge University Press}
}

@article{grosof2022optimal,
  title={Optimal scheduling in the multiserver-job model under heavy traffic},
  author={Grosof, Isaac and Scully, Ziv and Harchol-Balter, Mor and Scheller-Wolf, Alan},
  journal={Proceedings of the ACM on Measurement and Analysis of Computing Systems},
  volume={6},
  number={3},
  pages={1--32},
  year={2022},
  publisher={ACM New York, NY, USA}
}

@article{gupta2019load,
  title={Load balancing in the nondegenerate slowdown regime},
  author={Gupta, Varun and Walton, Neil},
  journal={Operations Research},
  volume={67},
  number={1},
  pages={281--294},
  year={2019},
  publisher={INFORMS}
}

@article{berg2017towards,
  title={Towards optimality in parallel scheduling},
  author={Berg, Benjamin and Dorsman, Jan-Pieter and Harchol-Balter, Mor},
  journal={Proceedings of the ACM on Measurement and Analysis of Computing Systems},
  volume={1},
  number={2},
  pages={1--30},
  year={2017},
  publisher={ACM New York, NY, USA}
}

@article{gamarnik2013steady,
  title={Steady-state {GI/GI/n} queue in the {Halfin-Whitt} regime},
  author={Gamarnik, David and Goldberg, David A},
  journal={The Annals of Applied Probability},
  pages={2382--2419},
  year={2013},
  publisher={JSTOR}
}

@article{berg2021hesrpt,
  title={{heSRPT}: Parallel scheduling to minimize mean slowdown},
  author={Berg, Benjamin and Vesilo, Rein and Harchol-Balter, Mor},
  journal={Performance Evaluation},
  volume={144},
  pages={102147},
  year={2020},
  publisher={Elsevier}
}

@inproceedings{scully2021gittins,
  title={The {Gittins} policy in the {M/G/1} queue},
  author={Scully, Ziv and Harchol-Balter, Mor},
  booktitle={2021 19th International Symposium on Modeling and Optimization in Mobile, Ad hoc, and Wireless Networks (WiOpt)},
  pages={1--8},
  year={2021},
  organization={IEEE}
}

@article{aalto2009gittins,
  title={On the {Gittins} index in the {M/G/1} queue},
  author={Aalto, Samuli and Ayesta, Urtzi and Righter, Rhonda},
  journal={Queueing Systems},
  volume={63},
  number={1--4},
  pages={437--458},
  year={2009},
  publisher={Springer}
}

@book{latouche1999introduction,
  title={Introduction to matrix analytic methods in stochastic modeling},
  author={Latouche, Guy and Ramaswami, Vaidyanathan},
  year={1999},
  publisher={SIAM}
}

@article{nino2001restless,
  title={Restless bandits, partial conservation laws and indexability},
  author={Ni{\~n}o-Mora, Jos{\'e}},
  journal={Advances in Applied Probability},
  volume={33},
  number={1},
  pages={76--98},
  year={2001},
  publisher={Cambridge University Press}
}

@article{gast2017refined,
  title={A refined mean field approximation},
  author={Gast, Nicolas and Van Houdt, Benny},
  journal={Proceedings of the ACM on Measurement and Analysis of Computing Systems},
  volume={1},
  number={2},
  pages={1--28},
  year={2017},
  publisher={ACM New York, NY, USA}
}

@inproceedings{leis2014morsel,
  title={Morsel-driven parallelism: a {NUMA}-aware query evaluation framework for the many-core age},
  author={Leis, Viktor and Boncz, Peter and Kemper, Alfons and Neumann, Thomas},
  booktitle={Proceedings of the 2014 ACM SIGMOD International Conference on Management of Data},
  pages={743--754},
  year={2014}
}

@inproceedings{delgado2018kairos,
  title={Kairos: Preemptive data center scheduling without runtime estimates},
  author={Delgado, Pamela and Didona, Diego and Dinu, Florin and Zwaenepoel, Willy},
  booktitle={Proceedings of the ACM Symposium on Cloud Computing},
  pages={135--148},
  year={2018}
}

@article{grosof2024bounds,
  title={Bounds on {M/G/k} scheduling under moderate load improving on {SRPT}-k and tightening lower bounds},
  author={Grosof, Isaac and Wang, Ziyuan},
  journal={ACM SIGMETRICS Performance Evaluation Review},
  volume={52},
  number={2},
  pages={24--26},
  year={2024},
  publisher={ACM New York, NY, USA}
}

@article{chen2025improving,
  title={Improving multiresource job scheduling with {Markovian} service rate policies},
  author={Chen, Zhongrui and Grosof, Isaac and Berg, Benjamin},
  journal={Proceedings of the ACM on Measurement and Analysis of Computing Systems},
  volume={9},
  number={2},
  pages={1--36},
  year={2025},
  publisher={ACM New York, NY, USA}
}

@article{grosof2019srpt,
title = {{SRPT} for multiserver systems},
journal = {Performance Evaluation},
volume = {127-128},
pages = {154-175},
year = {2018},
issn = {0166-5316},
author = {Isaac Grosof and Ziv Scully and Mor Harchol-Balter},
}

@inproceedings{tumanov2016tetrisched,
  title={TetriSched: global rescheduling with adaptive plan-ahead in dynamic heterogeneous clusters},
  author={Tumanov, Alexey and Zhu, Timothy and Park, Jun Woo and Kozuch, Michael A and Harchol-Balter, Mor and Ganger, Gregory R},
  booktitle={Proceedings of the Eleventh European Conference on Computer Systems},
  pages={1--16},
  year={2016}
}

@incollection{berg2021optimal,
  title={Optimal scheduling of parallel jobs with unknown service requirements},
  author={Berg, Benjamin and Harchol-Balter, Mor},
  booktitle={Handbook of Research on Methodologies and Applications of Supercomputing},
  pages={18--40},
  year={2021},
  publisher={IGI Global}
}

@article{gast2012mean,
  title={Mean field for {Markov} decision processes: from discrete to continuous optimization},
  author={Gast, Nicolas and Gaujal, Bruno and Le Boudec, Jean-Yves},
  journal={IEEE Transactions on Automatic Control},
  volume={57},
  number={9},
  pages={2266--2280},
  year={2012},
  publisher={IEEE}
}

@article{gast2010mean,
  title={A mean field model of work stealing in large-scale systems},
  author={Gast, Nicolas and Gaujal, Bruno},
  journal={ACM SIGMETRICS Performance Evaluation Review},
  volume={38},
  number={1},
  pages={13--24},
  year={2010},
  publisher={ACM New York, NY, USA}
}

@misc{li2025improvinggeneralizedcmurule,
      title={Improving Upon the generalized c-mu rule: a {Whittle} approach}, 
      author={Zhouzi Li and Keerthana Gurushankar and Mor Harchol-Balter and Alan Scheller-Wolf},
      year={2025},
      eprint={2504.10622},
      archivePrefix={arXiv},
}

@inproceedings{edmonds1999scheduling,
  title={Scheduling in the dark},
  author={Edmonds, Jeff},
  booktitle={Proceedings of the Thirty-First Annual ACM Symposium on Theory of Computing},
  pages={179--188},
  year={1999}
}

@inproceedings{jayaram2023sia,
  title={Sia: Heterogeneity-aware, goodput-optimized {ML}-cluster scheduling},
  author={Jayaram Subramanya, Suhas and Arfeen, Daiyaan and Lin, Shouxu and Qiao, Aurick and Jia, Zhihao and Ganger, Gregory R},
  booktitle={Proceedings of the 29th Symposium on Operating Systems Principles},
  pages={642--657},
  year={2023}
}

@inproceedings{qiao2021pollux,
  title={Pollux: Co-adaptive cluster scheduling for goodput-optimized deep learning},
  author={Qiao, Aurick and Choe, Sang Keun and Subramanya, Suhas Jayaram and Neiswanger, Willie and Ho, Qirong and Zhang, Hao and Ganger, Gregory R and Xing, Eric P},
  booktitle={15th USENIX Symposium on Operating Systems Design and Implementation},
  year={2021}
}

@inproceedings{wagner2021self,
  title={Self-tuning query scheduling for analytical workloads},
  author={Wagner, Benjamin and Kohn, Andr{\'e} and Neumann, Thomas},
  booktitle={Proceedings of the 2021 International Conference on Management of Data},
  pages={1879--1891},
  year={2021}
}

@inproceedings{lin2018model,
  title={A model-based approach to streamlining distributed training for asynchronous {SGD}},
  author={Lin, Sung-Han and Paolieri, Marco and Chou, Cheng-Fu and Golubchik, Leana},
  booktitle={2018 IEEE 26th International Symposium on Modeling, Analysis, and Simulation of Computer and Telecommunication Systems (MASCOTS)},
  pages={306--318},
  year={2018},
  organization={IEEE}
}

@article{Dai_2010,
   title={Many-server diffusion limits for {G/Ph/n+GI} queues},
   volume={20},
   ISSN={1050-5164},
   number={5},
   journal={The Annals of Applied Probability},
   publisher={Institute of Mathematical Statistics},
   author={Dai, J. G. and He, Shuangchi and Tezcan, Tolga},
   year={2010},
   month=oct }

@article{kurtz1970solutions,
  title={Solutions of ordinary differential equations as limits of pure jump {Markov} processes},
  author={Kurtz, Thomas G},
  journal={Journal of Applied Probability},
  volume={7},
  number={1},
  pages={49--58},
  year={1970},
  publisher={Cambridge University Press}
}

@article{Whittle_2005, title={Tax problems in the undiscounted case}, volume={42}, number={3}, journal={Journal of Applied Probability}, author={Whittle, Peter}, year={2005}, pages={754–765}}

@article{kaspi2011law,
  title={Law of large numbers limits for many-server queues},
  author={Kaspi, Haya and Ramanan, Kavita},
  journal={The Annals of Applied Probability},
  volume={21},
  number={1},
  pages={33--114},
  year={2011}
}

@inproceedings{larranaga2014index,
  title={Index policies for a multi-class queue with convex holding cost and abandonments},
  author={Larra{\~n}aga, Maialen and Ayesta, Urtzi and Verloop, Ina Maria},
  booktitle={The 2014 ACM International Conference on Measurement and Modeling of Computer Systems},
  pages={125--137},
  year={2014}
}

@book{asmussen2003applied,
  title={Applied probability and queues},
  author={Asmussen, S{\o}ren},
  year={2003},
  publisher={Springer}
}

@article{berg2022case,
  title={The case for phase-aware scheduling of parallelizable jobs},
  author={Berg, Benjamin and Whitehouse, Justin and Moseley, Benjamin and Wang, Weina and Harchol-Balter, Mor},
  journal={Performance Evaluation},
  volume={153},
  pages={102246},
  year={2022},
  publisher={Elsevier}
}

@inproceedings{berg2020optimal,
  title={Optimal resource allocation for elastic and inelastic jobs},
  author={Berg, Benjamin and Harchol-Balter, Mor and Moseley, Benjamin and Wang, Weina and Whitehouse, Justin},
  booktitle={Proceedings of the 32nd ACM Symposium on Parallelism in Algorithms and Architectures},
  pages={75--87},
  year={2020}
}

@article{zubeldia2020delay,
  title={Delay-Optimal Policies in Partial Fork-Join Systems with Redundancy and Random Slowdowns},
  author={Zubeldia, Martin},
  journal={Proceedings of the ACM on Measurement and Analysis of Computing Systems},
  volume={4},
  number={1},
  articleno={2},
  numpages={69},
  year={2020}
}

@inproceedings{eshwar2024online,
  title={Online Learning of Weakly Coupled {MDP} Policies for Load Balancing and Auto Scaling},
  author={Eshwar, S. R. and Felipe, Lucas Lopes and Reiffers-Masson, Alexandre and Menasch{\'e}, Daniel Sadoc and Thoppe, Gugan},
  booktitle={2024 {IFIP} Networking Conference ({IFIP} Networking)},
  pages={496--501},
  year={2024}
}

@article{gast2024lp,
  title={Linear Program-Based Policies for Restless Bandits: Necessary and Sufficient Conditions for (Exponentially Fast) Asymptotic Optimality},
  author={Gast, Nicolas and Gaujal, Bruno and Yan, Chen},
  journal={Mathematics of Operations Research},
  volume={49},
  number={4},
  pages={2468--2491},
  year={2024}
}

@misc{hong2024exponential,
  title={Achieving Exponential Asymptotic Optimality in Average-Reward Restless Bandits without Global Attractor Assumption},
  author={Hong, Yige and Xie, Qiaomin and Chen, Yudong and Wang, Weina},
  year={2024},
  eprint={2405.17882},
  archivePrefix={arXiv},
  primaryClass={cs.LG}
}

\appendix

\section{Convex optimization Translation}
\label{app:cvx}

In this appendix, we rewrite the optimization problem \eqref{eq:online opt} into a convex optimization problem. For completeness, we state the optimization again (up to some constant in the objective) as below:
\begin{equation}
        \begin{aligned}
        & \underset{k_i} {\text{minimize}}
        & &
        \sum_{i=1}^{\numclass}\frac{\rho_i c_i}{s_i(k_i)}
        \\
        & \text{subject to}
        & &
        \sum_{i=1}^{\numclass}\frac{\rho_i k_i}{s_i(k_i)}\leq 1
        ,\\
        &
        & & k_i\geq 1.
        \end{aligned}
\end{equation}

This is not immediately a convex optimization problem because the constraint terms $\frac{k_i}{s_i(k_i)}$ are not necessarily convex. However, by applying a change of variables, we can translate it into a convex optimization problem.

\begin{theorem}
    The optimization problem \eqref{eq:online opt} can be solved in the following two steps:
    \begin{enumerate}
        \item Solve the convex optimization problem \eqref{eq:convex} to get the optimal solution $z_i$;
        \item Get the optimal $k_i=s_i^{-1}(\frac{1}{z_i})$ where $s_i^{-1}$ is the inverse of the speedup function $s_i$.
    \end{enumerate}
    \label{thm:convex}
\end{theorem}
\begin{proof}

Define $z_i:=\frac{1}{s_i(k_i)}$. Define $d_i:=\sup_{k\geq 1} s_i(k)$ (if it does not exist, let $d_i=\infty$). If exists $k$ such that $s_i(k)=d_i$, denote the smallest one by $\xi_i$; Otherwise, let $\xi_i$ be $\perp$.

Note that since $s_i$ is non-decreasing and concave, we have that it is strictly increasing in $[1,\xi_i]$ ($[1,\infty)$ if $\xi_i=\perp$).
Thus we can define the function $\beta_i = s_i^{-1}$ to be the inverse of the speedup function on $[1,d_i]$ ($[1, d_i)$ if $\xi_i=\perp$).
Using the fact that $s_i$ is increasing, positive and concave, we have that $\beta_i$ is increasing and convex.

By definition of $z_i$, we have that $k_i=\beta_i(\frac{1}{z_i})$. Substituting this into the optimization problem \eqref{eq:online opt}, we have
\begin{equation}
    \begin{aligned}
    & \underset{z_i} {\text{minimize}}
    & & \sum_{i=1}^{\numclass}\rho_i c_i z_i\\
    & \text{subject to}
    & & \sum_{i=1}^{\numclass} \rho_i z_i\beta_i(\frac{1}{z_i})\leq 1\\
    & & & 1\geq z_i \geq \frac{1}{d_i}.\\
    \end{aligned}
    \label{eq:convex}
\end{equation}

Note that \[\frac{d}{dz_i}(z_i\beta_i(\frac{1}{z_i}))=\beta_i(\frac{1}{z_i}) - \frac{1}{z_i}\beta_i'(\frac{1}{z_i}),\]
and
\begin{align*}
    \frac{d^2}{dz_i^2}(z_i\beta_i(\frac{1}{z_i}))&=\frac{-1}{z_i^2}\beta_i'(\frac{1}{z_i}) +\frac{1}{z_i^2}\beta_i'(\frac{1}{z_i}) + \frac{1}{z_i^2}\beta_i''(\frac{1}{z_i})\\
    &=\frac{1}{z_i^2}\beta_i''(\frac{1}{z_i})>0 .
\end{align*}
Thus the optimization problem (\ref{eq:convex}) is a convex optimization, which we can numerically solve for the optimal $z_i$.

Finally, given the optimal $z_i$, we can get the optimal $k_i$ by letting $k_i=\beta_i(\frac{1}{z_i})$.

\end{proof}

\section{Handling the corner case of infinite optimal $k_i$}
\label{app:infinite k}

The optimization problem~\eqref{eq:online opt} may, in principle, fail to
attain its infimum: there exist speedup functions $s_i$ for which sending
$k_i\to\infty$ strictly decreases the objective without violating the budget.
Such speedups grow nearly linearly in $k$, in the sense that
$\lim_{k\to\infty} k/s_i(k) < \infty$. Realistic speedup functions do not
exhibit such linear growth --- the marginal benefit of an extra core always
diminishes as more cores are added --- so this regime arises only as a
theoretical corner case. This appendix treats the case for completeness and
shows that all results in the main text continue to hold after a single
preprocessing step.

Throughout, define
\[
L_i := \lim_{k\to\infty} \frac{k}{s_i(k)} \in (0,\infty].
\]

\paragraph{High-level approach}
The optimization~\eqref{eq:online opt} naturally partitions the classes
into two groups: a set $\mathcal{F}$ of classes whose optimal allocation
is finite ($k_i^* < \infty$), and a set $\mathcal{I}$ of classes whose
optimum is attained only as $k_i \to \infty$. Every class in $\mathcal{I}$
necessarily satisfies $L_i < \infty$ --- this is what allows
$k_i \to \infty$ without violating the budget. The main text analyzes the
case $\mathcal{I} = \emptyset$, and this appendix reduces the general case
to it: we set aside a dedicated core pool for each $i\in\mathcal{I}$ of
just enough size to drive that class's response time to zero in mean
field, and apply \policyone to the remaining classes $\mathcal{F}$ with a
correspondingly reduced budget.

\subsection{Identifying the corner-case classes}
\label{app:infinite k:partition}

We use the change of variables $z_i = 1/s_i(k_i)$ from
\ref{app:cvx}. The optimization~\eqref{eq:online opt} is
equivalent to the convex program~\eqref{eq:convex}, with $z_i$ ranging
over the compact interval $[0,1]$ (where $z_i = 0$ corresponds to
$k_i = \infty$, and the budget term is extended by continuity to
$z_i \beta_i(1/z_i)\big|_{z_i = 0} := L_i$).

Since the feasible set in $z$-coordinates is compact and the constraint
$\sum_i \lambda_i \E{X_i} z_i \beta_i(1/z_i) \leq \numcore$ is strictly convex (proved
in \ref{app:cvx}), the optimization has a \emph{unique} optimum
$z^* = (z_i^*)$. Define
\[
\mathcal{I} := \{i : z_i^* = 0\},
\qquad
\mathcal{F} := \{i : z_i^* > 0\}.
\]
By the above, $\mathcal{I}$ and $\mathcal{F}$ are uniquely determined by
the problem data. Classes in $\mathcal{I}$ are the corner-case classes;
classes in $\mathcal{F}$ have finite optimal allocation
$k_i^* = \beta_i(1/z_i^*) < \infty$.

\subsection{Extending \policyone to the corner case}
\label{app:infinite k:policy}

We now describe the modification of \policyone and verify that it
remains mean-field optimal.

Fix any $\beta\in(3/4,1)$, matching the Sub-Halfin-Whitt scaling used in
~\ref{app:verification}. In the $\systemindex^{th}$ system, for
each corner-case class $i\in\mathcal{I}$, reserve a dedicated pool of
\[
n_i\inasym \;:=\; \lambda_i \E{X_i} L_i \cdot \systemindex \;+\; (\lambda_i \E{X_i} L_i \cdot \systemindex)^{\beta}
\]
cores. Within each pool, class-$i$ jobs are served FCFS, and a job in
service is allocated all $n_i\inasym$ cores of its pool (so no two
class-$i$ jobs ever run concurrently). The remaining
$\numcore\inasym - \sum_{i\in\mathcal{I}} n_i\inasym$ cores are used to
run \policyone on the restricted class set $\mathcal{F}$, treating the
reduced budget as the new total.

We argue that this extension is mean-field optimal by decomposing the
weighted mean response time as
\[
\HC\inasym = \frac{1}{\lambda}\sum_{i\in\mathcal{I}} \lambda_i c_i \E{T_i\inasym} \;+\;
      \frac{1}{\lambda}\sum_{i\in\mathcal{F}} \lambda_i c_i \E{T_i\inasym}
\]
and treating the two sums separately.

For the $\mathcal{I}$ term, fix $i\in\mathcal{I}$. The class-$i$ pool is
a GI/G/1 system: jobs arrive at rate $\lambda_i\inasym = \systemindex\lambda_i$,
and each job of size $X_i$ requires service time
$S_i\inasym := X_i/s_i(n_i\inasym)$, so $\E{S_i\inasym} = \E{X_i}/s_i(n_i\inasym) = \Theta(\systemindex^{-1})$
and $\operatorname{Var}(S_i\inasym) = \operatorname{Var}(X_i)/s_i(n_i\inasym)^2 = \Theta(\systemindex^{-2})$.
The load is
\[
\rho_i\inasym
\;=\; \lambda_i\inasym\,\E{S_i\inasym}
\;=\; \frac{\systemindex \rho_i}{s_i(n_i\inasym)}
\;=\; \frac{\systemindex \rho_i \, L_i}{n_i\inasym}\bigl(1+o(1)\bigr)
\;=\; \frac{1 + o(1)}{1 + (\rho_i L_i \systemindex)^{\beta-1}}
\;=\; 1 - \Theta(\systemindex^{\beta-1}),
\]
where the third equality uses $n_i\inasym/s_i(n_i\inasym)\to L_i$ as
$\systemindex\to\infty$, and the fourth substitutes
$n_i\inasym = \rho_i L_i \systemindex + (\rho_i L_i \systemindex)^\beta$.
So $1-\rho_i\inasym = \Theta(\systemindex^{\beta-1})$.
By Theorem 2 of \cite{kingman1962inequalities}, the mean queueing time
satisfies
\[
\E{W_q\inasym}
\;\leq\; \frac{\lambda_i\inasym\bigl(\operatorname{Var}(A_i\inasym)+\operatorname{Var}(S_i\inasym)\bigr)}{2(1-\rho_i\inasym)}
\;=\; \Theta(\systemindex^{-\beta})
\;\xrightarrow{\systemindex\to\infty}\; 0,
\]
using $\operatorname{Var}(A_i\inasym) = \operatorname{Var}(A_i)/\systemindex^2 = \Theta(\systemindex^{-2})$.
Combined with $\E{S_i\inasym}\to 0$, the class-$i$ mean response time
$\E{T_i\inasym} = \E{W_q\inasym} + \E{S_i\inasym}$ vanishes, so the
$\mathcal{I}$ term in the decomposition above contributes $0$ in the
limit.

For the $\mathcal{F}$ term, observe that the reduced budget
$\numcore\inasym - \sum_{i\in\mathcal{I}} n_i\inasym
= \bigl(\numcore - \sum_{i\in\mathcal{I}}\lambda_i \E{X_i} L_i\bigr)\systemindex - o(\systemindex)$
scales linearly in $\systemindex$, and the corresponding optimization on
$\mathcal{F}$ has all $k_i^* < \infty$ by construction. Mean-field
optimality of \policyone on this reduced problem then follows by direct
application of the main text.

Finally, in the limit $\systemindex\to\infty$, these two contributions
together match the infimum of~\eqref{eq:online opt} (which,
via~\eqref{eq:budget is lower bound}, extends to a lower bound on the
\ourproblem problem even in the corner case): the $\mathcal{I}$-classes
contribute $0$, and the $\mathcal{F}$-classes contribute exactly the
optimal value of the reduced optimization. Hence the extended policy is
mean-field optimal.

\section{Almost no parallelism happens in asymptotic regimes that $\rho\to 1$}
\label{appendix:asym}



The following proposition shows that for the \ourproblem problem with strictly sub-linear speedup functions under asymptotic regimes where $\rho\to1$, the fraction of work done in parallelism goes to 0. Most notation is defined in Section~\ref{sec:setting}.

\begin{proposition}
\label{lemma: app asym}
Define $\rho\inasym:=\sum_{i=1}^\numclass \lambda_i\inasym \E{X_i}/\numcore\inasym$ to be the total load in the $\systemindex^{th}$ system. Assume that the speedup functions of all classes are strictly sub-linear and the asymptotic regime has $\lim_{\systemindex\to \infty}\rho^{(\systemindex)}\to 1$. 
For any $k_0>1$, further define $q\inasym(k_0)$ to be the time-average proportion of inherent work finished using more than $k_0$ cores. Then 
    \[\forall k_0>1, \lim_{\systemindex\to \infty} q\inasym(k_0)\to 0.\]    
\end{proposition}

\begin{proof}
Recall that each type-$i$ job brings $X_i$ inherent work to the system. Thus we have that the time average inflowing rate of the inherent work is
\begin{equation}
    \sum_{i=1}^{\numclass} \lambda_i\inasym \cdot \E{X_i}=  \rho\inasym \cdot \numcore\inasym.
\end{equation}

    If a class-$i$ job is running on $k\geq k_0$ cores, then a unit inherent work is finished in $\frac{1}{s_i(k)}$ units of time. Thus the {\em core-second} needed for a unit inherent work is $k\cdot \frac{1}{s_i(k)}> \frac{k_0}{s_i(k_0)}$ (because of the concavity of the speedup function $s_i$). 
    
    Thus, for the  $q\inasym(k_0)$ proportion of inherent work finished using more than $k_0$ cores, the total core-second needed is at least 
    \[q\inasym(k_0) \cdot \rho\inasym \numcore\inasym\cdot \frac{k_0}{\max_i s_i(k_0)}.\]

    Moreover, for the rest $1-q\inasym(k_0)$ proportion of inherent work, by sub-linearity of the speedup functions, the total core-second needed is at least $(1-q\inasym(k_0))\cdot\rho\inasym\numcore\inasym$.
    Note that for the system to be stable, the time average total core-second needed must be no more than $\numcore\inasym$. Thus we have 
    \[q\inasym(k_0) \cdot \rho\inasym\numcore\inasym \cdot \frac{k_0}{s_i(k_0)} + (1-q\inasym(k_0))\cdot\rho\inasym\numcore\inasym\leq \numcore\inasym.\]

    This inequality gives that 
    \[q\inasym(k_0)\leq \left(\frac{1}{\rho\inasym}-1\right)\cdot\left(\frac{1}{\frac{k_0}{\max_i s_i(k_0)}-1}\right).\]
    Since $\lim_{\systemindex\to \infty}\rho^{(\systemindex)}\to 1$, we have that for any $k_0>1$, $\lim_{\systemindex\to \infty} q\inasym(k_0)\to 0$.

\end{proof}

\section{Mean field optimal policies should not use exact job size information}
\label{app:no job size}
In this section, we prove that for any policy that is mean field optimal, it almost surely allocates the same number of cores to all class-$i$ jobs, regardless of their exact job size. The proof uses results and notation from Section~\ref{sec:relax}.

\begin{proposition}
\label{prop:app no job size}
Assume the speedup function for any class $i$ is strictly concave.
Define $k_i:=\kR{\lambda_i}{\numcore}$.
    For any policy $\pi$ that is mean field optimal, and for any class $i$,  we have that the fraction of class-$i$ inherent work completed using $k_i$ cores goes to 1 in mean field.
\end{proposition}
\begin{proof}
    We prove the proposition by contradiction. Suppose the statement does not hold, then either (1) the fraction of class-$i$ work completed using $k'\neq k_i$ goes to 1 in mean field, or (2) the fraction of class-$i$ work is distributed among multiple values (or a range) of $k$, such that the fraction does not converge to 1 for any single core count. 

    We first prove that (2) cannot hold for any class $i$. Suppose (2) holds for some class $i$. Let $h_i(k)$ be the fraction of time that a class-$i$ job is running on $k$ cores ($\int_{1}^\infty h_i(k)dk=1$). 
    Now we create a policy $\pi'$ for the \budget problem, which always allocates $k':=\int_{1}^\infty k\cdot h_i(k)dk$ cores to class-$i$ jobs. By concavity of the speedup function, we have that the mean response time of class-$i$ job under $\pi'$ is no worse than that under $\pi$ (see Lemma 1 in \cite{li2024rentgpusbudget}). Moreover, the time-average number of cores used under $\pi'$ is strictly better than that under $\pi$: the time-average number of cores used under $\pi'$ is $\frac{k' \lambda_i\E{X_i}}{s_i(k')}$, and that under $\pi$ is 
    \begin{align*}
        \int_{0}^{k'} \frac{k \lambda_i\E{X_i} \cdot \frac{s_i(k)h_i(k)}{\int_{1}^\infty s_i(x)h_i(x) dx}}{s_i(k)} dk & = \int_{0}^{k'} k \lambda_i\E{X_i} \cdot \frac{h_i(k)}{\int_{1}^\infty s_i(x)h_i(x) dx} dk \\
        &= \frac{1}{\int_1^\infty s_i(x)h_i(x) dx} \lambda_i \E{X_i} k'\\
        &> \frac{\lambda_i\E{X_i}k'}{s_i(k')}.
    \end{align*}
    Define the reduced ratio of the time-average number of cores used to be some $c>0$. 
    Then, we have that
    \begin{align*}
        \lim_{\systemindex\to\infty}\HC_{\ourproblem\inasym}^\pi &= \lim_{\systemindex\to\infty}\HC_{\budget\inasym}^\pi \\
        &\geq \lim_{\systemindex\to\infty}\HC_{\budget\inasym}^{\pi'}.
    \end{align*}

    Note that $\pi'$ only uses $(1-c)n$ cores on average in the Relaxed problem. By \eqref{eq:online opt}, the optimal weighted mean response time is strictly larger than the optimal weighted mean response time of the Relaxed policy using on average $n$ cores because the resource constraint shrank by a $(1-c)$ factor. Hence the weighted mean response time under $\pi$ is strictly larger than that of the optimal Relaxed problem, hence is sub-optimal.

    Next, since (2) does not hold for any class $i$, under $\pi$ each class $i$ job only runs on a single number of cores in mean field, denoted by $k_i'$. This means that $k_i'$ is a feasible solution of the optimization problem~\eqref{eq:online opt}. Since the speedup functions are strictly concave, the optimization problem~\ref{eq:online opt} can be translated into a strictly convex optimization and the optimal solution is unique. Thus, $k_i'=k_i$.

\end{proof}

\section{Proof for Lemma~\ref{lemma: k converges}}
\label{app: k converge}
In this section, we prove Lemma~\ref{lemma: k converges}, which is listed below.

\begin{lemma}
    For any class $i$,
    $\lim_{\systemindex\to\infty}\kR{\lambda\inasym}{\numcore\inasym - {\left(\numcore\inasym\right)}^\beta}\to \kR{\lambda}{\numcore}$.
\end{lemma}
\begin{proof}
    By definition, $\kR{\lambda\inasym}{\numcore\inasym - {\left(\numcore\inasym\right)}^\beta}$ is the optimal solution of the following optimization problem:
    \begin{equation}
        \begin{aligned}
        & \underset{k_i} {\text{minimize}}
        & & 
        \sum_{i=1}^\numclass \frac{\lambda_i\cdot \systemindex}{\lambda\cdot \systemindex}\frac{ \E{X_i} c_i }{s_i(k_i)}
        \\
        & \text{subject to}
        & & 
        \sum_{i=1}^\numclass\frac{\lambda_i\cdot \systemindex \E{X_i} k_i}{s_i(k_i)}\leq \numcore\cdot \systemindex - (\numcore\cdot \systemindex)^\beta .\\
        \end{aligned}
\end{equation}

Dividing by $\systemindex$ on both sides of the constraint, we have that the optimization problem is equivalent to:
\begin{equation}
        \begin{aligned}
        & \underset{k_i} {\text{minimize}}
        & & 
        \sum_{i=1}^\numclass \frac{\lambda_i}{\lambda}\frac{ \E{X_i} c_i }{s_i(k_i)}
        \\
        & \text{subject to}
        & & 
        \sum_{i=1}^\numclass\frac{\lambda_i\E{X_i} k_i}{s_i(k_i)}\leq \numcore - \numcore^\beta\cdot \systemindex^{\beta-1} .\\
        \end{aligned}
        \label{eq:optim in app}
\end{equation}

Note that both the objective function and constraint function are continuous in $k_i$. Moreover, the optimization problem is a convex optimization problem (see \ref{app:cvx}). Thus, since $\lim_{\systemindex\to \infty} d^{\beta-1}\to 0$, we have that both the solution and the value of the optimization problem \eqref{eq:optim in app} converge to those of the optimization problem~\eqref{eq:online opt}.

\end{proof}

\section{Verification of the conditions in \cite{li2025simple}}
\label{app:verification}
Here, we state Corollary 3 from \cite{li2025simple} in its original form.
\begin{lemma}[From Corollary 3 in \cite{li2025simple}]
For a GI/GI/$\numcore$ queue with inter-arrival time $A$ and job size distribution $S$, suppose: {\em (i)} $\E{A^2}<\infty$, {\em (ii)} $\E{S^3}<\infty$, {\em (iii)} $\frac{1}{\E{A}}<\numcore\frac{1}{\E{S}}$, {\em (iv)} $\numcore(1-\rho)^2\geq 10^6\left(\E{(\frac{A}{\E{A}})^2}\right)^2$, and {\em (v)}
a steady-state distribution of number of jobs in the system exists, then 
\[\E{\text{steady-state number of queueing jobs}}\leq \frac{a_1+a_2}{1-\rho}\cdot 1000(\numcore(1-\rho)^2)^{-0.5},\]
where $a_1,a_2$ are constant with respect to $\numcore$.
\label{lemma:goldberg:full}
\end{lemma}

To apply Lemma \ref{lemma:goldberg:full}, we will show that the conditions for the Lemma are satisfied in a GI/GI/${\numcore_i'}\inasym$ system for sufficiently large $\systemindex$.
Specifically, we note that condition (iv) in the Lemma will be satisfied for sufficiently large $\systemindex$ if the system is in a Sub-Halfin-Whitt scaling regime.  
Furthermore, the bound on the number of jobs in queue will converge to 0 if the system is in a strongly Sub-Halfin-Whitt regime where the system load scales as $1-\omega\left(\left(\numcore\inasym\right)^{-.25}\right)$.
For our purposes in this paper, this understanding of the Corollary is sufficient to prove the optimality of \policyone.
However, \cite{li2025simple} states the result in the above, more general form to give additional information about convergence rates and scaling behavior.
We now verify the conditions for the Lemma below.

As is also assumed in \cite{li2025simple}, we assume that the steady-state distribution of the number of jobs in system exists.
Hence, (v) is satisfied by assumption.
We now verify that conditions (i) - (iv) hold for sufficiently large $\systemindex$. 

\begin{enumerate}
    \item[(i)] $\E{A_i^2}<\infty$: In the $\systemindex^{th}$ system, we have that $\E{(A_i\inasym)^2} = \frac{\E{A_i^2}}{d^2}$, which is finite for any $\systemindex$.
    \item[(ii)] $\E{S^3}<\infty$: In the $\systemindex^{th}$ system, we have that $\E{\left(S_i\inasym\right)^3}=\E{\left(\frac{X_i}{s_i(k_i\inasym)}\right)^3}$, which is finite for any $\systemindex$. 
    \item[(iii)] $\frac{1}{\E{A}}<\numcore\frac{1}{\E{S}}$: In the $\systemindex^{th}$ system, define $\rho_i\inasym:= \frac{\lambda_i\inasym\E{S_i\inasym}}{\left(n_i'\right)\inasym}$. Then, this inequality is equivalent to $\rho_i\inasym<1$, and it suffices to verify the inequality~\eqref{eq: FW CAM proof eq 6} below.
    \item[(iv)] $\numcore(1-\rho)^2\geq 10^6\left(\E{(\frac{A}{\E{A}})^2}\right)^2$: In the $\systemindex^{th}$ system, this inequality is equivalent to 
    \begin{equation}
        \left(n_i'\right)\inasym (1-\rho_i\inasym)^2 \geq 10^6 \left(\frac{\E{\left(A_i\inasym\right)^2}}{\E{A_i\inasym}^2}\right)^2.
        \label{eq: FW CAM proof eq 6}
    \end{equation}
    Since we only care about systems when $\systemindex\to\infty$, we can simplify the terms by considering their orders with respect to $\systemindex$. 
By definition, we have $\lambda_i\inasym = \Theta(\systemindex)$.
By Lemma~\ref{lemma: k converges}, we have that $\E{S_i\inasym} = \Theta(1)$ and $k_i\inasym = \Theta(1)$. Thus we have that 

\begin{align}
    \left(n_i'\right)\inasym&= \lfloor\numcore\inasym \cdot \frac{r_i\inasym}{\sum_{j=1}^\numclass r_j\inasym} \cdot \frac{1}{k_i\inasym}\rfloor \nonumber\\
    &=\numcore\inasym \cdot \frac{r_i\inasym}{\sum_{j=1}^\numclass r_j\inasym} \cdot \frac{1}{k_i\inasym} + O(1)\nonumber \\
    &= \numcore\inasym \cdot \frac{\lambda_i\inasym \E{S_i\inasym}}{\sum_{j=1}^\numclass r_j\inasym} + O(1)\nonumber\\
    &\geq\lambda_i\inasym \E{S_i\inasym} \cdot \frac{\numcore\inasym}{\numcore\inasym-\left(\numcore\inasym\right)^\beta} + O(1) && \text{by \eqref{eq:FW CAM total effective load}} \label{eq:ni' bound}\\
    &=\Theta(d).
\end{align}

Moreover, we have that $\left(n_i'\right)\inasym \leq n\inasym =\Theta(d)$. Hence,
\begin{equation}
    \left(n_i'\right)\inasym =\Theta(d).
    \label{eq: ni' order}
\end{equation}

We can now analyze the order of $1-\rho_i\inasym$:

\begin{align*}
    1-\rho_i\inasym &= \frac{\left(n_i'\right)\inasym - \lambda_i\inasym\E{S_i\inasym}}{\left(n_i'\right)\inasym} \\
    &\geq \frac{\lambda_i\inasym \E{S_i\inasym} \cdot \frac{\numcore\inasym}{\numcore\inasym-\left(\numcore\inasym\right)^\beta} + O(1) - \lambda_i\inasym \E{S_i\inasym}}{\Theta(d)}\\
    &= \frac{1}{\Theta(d)}\left(\lambda_i\inasym \E{S_i\inasym}\frac{\left(\numcore\inasym\right)^\beta}{\numcore\inasym-\left(\numcore\inasym\right)^\beta}   +O(1)\right)\\
    &= \frac{1}{\Theta(d)}\left(\Theta(d)\frac{\Theta(d^\beta)}{\Theta(d)}   +O(1)\right)\\
    &=\Theta(d^{\beta-1})
    \end{align*}

Thus, the left hand side of inequality~\eqref{eq: FW CAM proof eq 6} is $\Theta(d^{2\beta-1})$ where $2\beta-1>0$, while the right hand side of \eqref{eq: FW CAM proof eq 6} is of order $\Theta(1)$ by definition. Hence the inequality always holds when $\systemindex\to \infty$.

\end{enumerate}

\section{Omitted proof for optimal policy of the single-job problem}
\label{app: proof for SAB}

In this appendix we prove the following omitted lemma in the proof of Lemma~\ref{lemma:single arm}.

\begin{lemma}
\label{lemma: app single arm}
    For any discount factor $\alpha\to 0$ and service penalty $\ell$, the optimal policy solving \eqref{eq:single} always chooses the same action throughout the job's lifetime. In other words,  the optimal policy either always chooses action $k=0$ and leaves the job forever incurring some cost, or keeps choosing some action $k\geq 1$ until the job is completed.
\end{lemma}
\begin{proof}
    First, if the optimal action at state $x$ is to choose $k=0$, since the state will not change, we have that the optimal policy will choose $k=0$ forever. 

    Otherwise, the lemma holds by concavity of the speedup function $s_i$. Mathematically, suppose the optimal policy $\pi$ chooses action $k(t)\geq 1$ at time $t$ for $t\in [0,t']$. One can create another policy $\pi'$, which use the same action throughout this $t'$ time, incurring the same service penalty but does more work. Mathematically, define
    \[k'=\frac{\alpha}{1-e^{-\alpha t'}}\int_{0}^{t'} k(t) e^{-\alpha t} dt.\]
    The policy $\pi'$ allocates $k'$ cores to the job throughout time $[0,t']$. Compared with $\pi$, the two policies incur exactly the same total discounted service penalty for any $\alpha$ and $\ell$. However, the total work $\pi'$ does is $s_i(k') t'$, while the total work $\pi$ does is $\int_{0}^{t'} s_i(k(t)) dt$. Note that when $\alpha\to 0$, $k'$ goes to the time-average value of $k(t)$ within time $[0,t']$. Thus, by strict concavity of the speedup function, we have that $\pi'$ is strictly better than $\pi$ when $\alpha\to 0$.

\end{proof}

We also prove the following lemma.
\begin{lemma}
    For any class $i$, the function $f_i(k):=\frac{s_i'(k)}{s_i(k)-ks_i'(k)}$ is strictly decreasing with $k$.
\label{lemma:f decreasing}
\end{lemma}
\begin{proof}
    We only need to show that $\frac{1}{f_i(k)}$ is increasing, which is $\frac{s_i(k)}{s_i'(k)} - k.$

    We can take the derivative of this expression:
    \begin{align*}
        \frac{d(\frac{1}{f_i(k)})}{dk} &= \frac{s_i'(k)^2 - s_i(k)s_i''(k)}{s_i'(k)^2} - 1\\
        &= \frac{-s_i(k)s_i''(k)}{s_i'(k)^2} \\
        &> 0.
    \end{align*}
    The last inequality holds because the speedup functions are (strictly) concave.
\end{proof}

\section{Proof for job sizes that follow phase-type distribution}
\label{app:ph type}

In this appendix, we present the proof for the mean field optimality of \policytwo under the assumption of phase-type distributed job sizes. The proof is similar to that presented in Section~\ref{sec: policy two proof}.

\begin{assumption}
\label{assumption:PH}
    Assume the job size distribution $X_i$ is a phase-type distribution defined by the following Markov chain: There are $K_i$ states. Let $\initp_i=(p_{iu})$ denote the vector of the initial probability to be in state $u$, and let $\transm_i$ be the sub-generator matrix of the Markov chain, i.e., $\transm_i$ is a $K_i\times K_i$ matrix, and $\transm_i(u,v)$ is the transition rate from $u$ to $v$ $u\neq v$, $\transm_i(u,u):=-\sum_{v\neq u} \transm_i(u,v) - q_i(u,0)$ where $q_i(u,0)$ is the transition rate from state $u$ to the absorbing state.  
\end{assumption}

Under the mean field regime, the limiting dynamics of the queueing system can be captured by an ODE~\cite{kurtz1970solutions,gast2017refined}. Notably, even for a GI arrival process, the first-order drift is the same as Poisson arrivals, thus it can be treated as if it is a Poisson arrival in the ODE (Theorem 3 in ~\cite{whittle1988restless}). Thus, the system under policy \policytwo in mean field can be described by the following ODE: 

\begin{align}
    &\dot{\massv}_i = \lambda_i \initp_i + \massv_i \transm_i \cdot s_i(g_i(\ell)), \label{eq:ODE} \\
    &\text{where $\ell$ is determined by }\sum_{i=1}^\numclass \massv_i \cdot \textbf{1}\cdot g_i(\ell) = \numcore,\text{ and $g_i$ is defined in Definition~\ref{def:policytwo}.}\nonumber
\end{align}
Here $\massv_i$ is the vector of the fluid-scaled number of class-$i$ jobs, and $\massv_i(u)$ is the fluid-scaled number of class-$i$ jobs in state $u$. For a better understanding of the ODE \eqref{eq:ODE}, $\lambda_i \initp_i$ is the arrival rate of class-$i$ jobs, $\massv_i \transm_i$ represents the rate of transition between the states of class-$i$ jobs (if running on 1 core), and $s_i(g_i(\ell))$ is the speedup that any class-$i$ job gets. Finally, by the definition of \policytwo, $\ell$ is picked such that the total number of cores used reaches $\numcore$, which is $\sum_{i=1}^\numclass \massv_i \cdot \textbf{1}\cdot g_i(\ell) = \numcore$.

We can now state our main theorem. Note that the asymptotic optimality of \policytwo relies on the assumption that ODE \eqref{eq:ODE} has a global attractor. Similar assumptions of the existence of a global attractor are also made in prior works (e.g., \cite{weber1990,Verloop_2016}, see \cite{Verloop_2016} for a detailed discussion).  

\begin{theorem}[mean field optimality of \policytwo]
    \label{thm:mean field optimality of policytwo}
    If ODE~\eqref{eq:ODE} has a global attractor, then \policytwo is mean field optimal. 
\end{theorem}
\begin{proof}
Suppose ODE~\eqref{eq:ODE} has a global attractor, then the global attractor must be a stationary point. We next prove that there exists a unique stationary point for ODE~\eqref{eq:ODE}, and the weighted mean response time for the stationary point is equal to the optimal weighted mean response time of the Relaxed problem (see \eqref{eq:online opt}). Note that if this statement holds, we have that \policytwo is mean field optimal because its weighted mean response time converges to the lower bound given by the Relaxed problem.

    For the stationary point $\massv_i$, since $\dot{\massv}_i=0$, we have that 
    \[\massv_i \transm_i \cdot s_i(g_i(\ell))= - \lambda_i \initp_i.\]
    Since matrix $\transm_i$ is invertible (Lemma 2.2.1 in \cite{latouche1999introduction}), we have that 
    \begin{equation}
        \massv_i = - \frac{\lambda_i}{s_i(g_i(\ell))} \initp_i\left(\transm_i\right)^{-1}.
        \label{eq:eq14}
    \end{equation}
    Moreover, we have that $\E{X_i} = -\initp_i\left(\transm_i\right)^{-1}\textbf{1}$ (Theorem 2.2.3 in \cite{latouche1999introduction}). Thus, by multiplying $\textbf{1}$ on both sides of \eqref{eq:eq14}, we have that 
    \[\massv_i\textbf{1} = \frac{\lambda_i}{s_i(g_i(\ell))}\E{X_i}.\]
    Now we can find the stationary point by solving the equation 
    \[\sum_{i=1}^\numclass \massv_i\textbf{1}\cdot g_i(\ell)=\numcore,\]
    which is equivalently
    \begin{equation}
        \sum_{i=1}^\numclass \frac{\lambda_i g_i(\ell)}{s_i(g_i(\ell))}\E{X_i} = n.
        \label{eq:eq15}
    \end{equation}

    Note that $g_i$ is strictly decreasing for any $\ell< c_if_i(1)$, also $\frac{k}{s_i(k)}$ is strictly increasing with $k$, thus we have that the left hand side is strictly decreasing with $\ell<\max_i c_if_i(1)$. Moreover, for any $\ell\geq \max_i c_if_i(1)$, we have $g_i(\ell)=1$ for any $i$, and the left hand side is equal to $\sum_{i}\lambda_i \E{X_i}<n$. This means there exists a unique $\ell^*$ solving \eqref{eq:eq15}. Then, the stationary point $\massv_i$ can be solved from \eqref{eq:eq14}.
    
    Since the dynamics of the system under \policytwo in mean field is captured by ODE~\eqref{eq:ODE} whose global attractor is $\massv_i$, we have that in the limit, any class-$i$ job is assigned $g_i(\ell^*)$ cores. Thus
    the mean response time of a class-$i$ job is $\frac{\E{X_i}}{s_i(g_i(\ell^*))}$, and the weighted mean response time in the limit is
    \[\lim_{\systemindex\to\infty} \HC_{\ourproblem\inasym}^{\policytwo} = \frac{1}{\lambda}\sum_{i=1}^\numclass \frac{c_i\lambda_i\E{X_i}}{s_i(g_i(\ell^*))}.\]

    One can use Lagrange multipliers to solve the optimization problem~\eqref{eq:online opt} and find that the optimal weighted mean response time for the Relaxed problem is $ \frac{1}{\lambda}\sum_{i=1}^\numclass \frac{c_i\lambda_i\E{X_i}}{s_i(g_i(\ell^*))}$ (Lemma~\ref{lemma:appendix D}). Thus we have that \policytwo is mean field optimal.

\end{proof}

\section{GREEDY is not mean field optimal}
\label{app:greedy not mf}

In Section~\ref{sec: sec 6} and Section~\ref{sec:simulation}, we
argued informally that GREEDY~\cite{berg2017towards} fails to be mean
field optimal. In this appendix, we make this precise via the ODE
method of Section~\ref{sec: policy two proof}: we write down the mean
field ODE for GREEDY, identify its stationary point, and observe that
it differs from the optimum identified in
Theorem~\ref{thm:mean field optimality of policytwo exp}. Similar methods apply to other heuristics in the
literature, such as EQUI~\cite{edmonds1999scheduling} and
KNEE~\cite{lin2018model}; we focus on GREEDY for concreteness.

We treat the case of exponential job sizes ($X_i\sim \Exp(\mu_i)$) and
unit weights ($c_i=1$); since GREEDY was designed for the mean
response time, this case alone is enough to exhibit its suboptimality.

\subsection{GREEDY's allocation in mean field}

In any state with $N_i$ class-$i$ jobs in the system, GREEDY gives
every class-$i$ job the same allocation $k_i$ (by symmetry and the
concavity of $s_i$), where $\{k_i\}$ solve
\[
    \max \sum_{i=1}^\numclass N_i\, s_i(k_i)
    \quad \text{subject to} \quad
    \sum_{i=1}^\numclass N_i k_i \leq \numcore.
\]
A standard marginal argument shows that the resulting per-job
allocation is
\begin{equation*}
    h_i(\eta) \;:=\;
    \begin{cases}
        1 & \text{if } s_i'(1)\leq \eta, \\
        (s_i')^{-1}(\eta) & \text{otherwise,}
    \end{cases}
\end{equation*}
where $\eta$ is the smallest ``marginal value of a core'' such that
the total allocation fits the budget. So GREEDY is itself a
"market-clearing" policy, but it equalizes $s_i'(k_i)$ across classes,
whereas \policytwo (Definition~\ref{def:policytwo}, under $c_i=1$)
equalizes $f_i(k_i)$. Since
these two equalizing conditions differ, GREEDY's per-class
allocations differ from \policytwo's, and as we now show this
discrepancy means GREEDY's ODE stationary point is not the mean field
optimum.

\subsection{The ODE and its stationary point}

In the mean field scaling regime ($\systemindex\to\infty$, $\rho$
fixed), let $z_i(t)$ be the fluid-scaled number of class-$i$ jobs in
the system. By the standard fluid
limit~\cite{kurtz1970solutions,gast2017refined} (using Theorem~3
of~\cite{whittle1988restless} for GI arrivals), GREEDY's mean field
dynamics are
\begin{align}
    \dot z_i \;=\; \lambda_i \;-\; z_i\, \mu_i\, s_i\!\bigl(h_i(\eta)\bigr),
    \qquad \text{where $\eta$ satisfies} \quad
    \sum_{i=1}^\numclass z_i\, h_i(\eta) = \numcore.
    \label{eq:ODE greedy}
\end{align}
This is exactly the \policytwo ODE~\eqref{eq:ODE exp} with $g_i(\ell)$
replaced by $h_i(\eta)$.

Setting $\dot z_i = 0$ yields the stationary point
\begin{equation}
    z_i^{\text{GREEDY}} \;=\; \frac{\lambda_i\,\E{X_i}}{s_i\bigl(h_i(\eta^{\text{GREEDY}})\bigr)},
    \qquad \text{where $\eta^{\text{GREEDY}}$ solves} \quad
    \sum_{i=1}^\numclass \frac{\lambda_i\,\E{X_i}\, h_i(\eta^{\text{GREEDY}})}{s_i\bigl(h_i(\eta^{\text{GREEDY}})\bigr)} \;=\; \numcore.
    \label{eq:greedy stationary}
\end{equation}
By the same monotonicity argument as in the proof of
Theorem~\ref{thm:mean field optimality of policytwo exp}, $\eta^{\text{GREEDY}}$
exists and is unique. Assuming it is a global attractor of
\eqref{eq:ODE greedy} (in analogy with the assumption made for
\policytwo), the mean response time of GREEDY in mean field is
\begin{equation}
    \lim_{\systemindex\to\infty} \HC_{\ourproblem\inasym}^{\text{GREEDY}}
    \;=\; \frac{1}{\lambda}\sum_{i=1}^\numclass \frac{\lambda_i\,\E{X_i}}{s_i\!\bigl(h_i(\eta^{\text{GREEDY}})\bigr)}.
    \label{eq:greedy cost}
\end{equation}

\subsection{Why GREEDY can be suboptimal in mean field}

The expressions \eqref{eq:greedy stationary}--\eqref{eq:greedy cost}
mirror those for \policytwo (equations
\eqref{eq:14 exp}--\eqref{eq:eq15 exp} and the cost expression in
Theorem~\ref{thm:mean field optimality of policytwo exp}); only the
per-class allocations differ. By Lemma~\ref{lemma:appendix D}, the
lower bound from the Relaxed problem~\eqref{eq:online opt} is
achieved exactly when every class-$i$ job gets $g_i(\ell^*)$ cores,
and this minimizer is unique by the strict convexity
of~\eqref{eq:online opt} (\ref{app:cvx}). Hence GREEDY
matches this lower bound only when $h_i(\eta^{\text{GREEDY}}) = g_i(\ell^*)$ for
every class. However, the GREEDY allocations $h_i(\eta^{\text{GREEDY}})$ do not
agree with the \policytwo allocations $g_i(\ell^*)$ in most cases. For example, in the
two-class workload evaluated in Section~\ref{sec:simulation}, the two allocations disagree and thus GREEDY is suboptimal in mean field, matching the gap
seen in Figure~\ref{fig:simulation_results}(b) and
Figure~\ref{fig:simulation_results_gen}.

\section{Solution of the optimization problem~\ref{eq:online opt}}

\begin{lemma}
\label{lemma:appendix D}
We have that the optimal value of the optimization~\eqref{eq:online opt} is given by 
    \[\frac{1}{\lambda}\sum_{i=1}^\numclass \frac{c_i\lambda_i\E{X_i}}{s_i(g_i(\ell^*))},\]
    where $\ell^*$ is defined by $\sum_{i=1}^\numclass \frac{\lambda_i g_i(\ell^*)}{s_i(g_i(\ell^*))}\E{X_i} = n.$
\end{lemma}
\begin{proof}
Using Lagrange multiplier method, we define 
\[L:=\sum_{i=1}^\numclass \frac{ \lambda_i\E{X_i} c_i }{s_i(k_i)} + \gamma \left(\cdot \sum_{i=1}^\numclass\frac{\lambda_i\E{X_i} k_i}{s_i(k_i)}-n\right).\]
Then by the KKT condition, we have that the optimal policy $k_i$ satisfies:  
\begin{align*}
    & \frac{\partial L}{\partial k_i}=0 \quad \text{or} \quad \big( k_i=1 \ \text{and}\ \frac{\partial L}{\partial k_i} \Big|_{k_i=1}\geq 0 \big)\\
    & \sum_{i=1}^\numclass\frac{\lambda_i\E{X_i} k_i}{s_i(k_i)} = n.
\end{align*}

Note that \[\frac{\partial L}{\partial k_i} = \lambda_i\E{X_i}\left(\frac{c_i\cdot (-1)s_i'(k_i)}{s_i(k_i)^2} + \gamma \frac{s_i(k_i) - k_i s_i'(k_i)}{s_i(k_i)^2}\right).\]

Thus, $\frac{\partial L}{\partial k_i}=0$ is equivalent to 
\[\frac{\gamma}{c_i} = \frac{s_i'(k_i)}{s_i(k_i)-k_i s_i'(k_i)} = f_i(k_i),\]
and $k_i=1$ if and only if $f_i(1)\leq \frac{\gamma}{c_i}$. Thus we have $k_i=g_i(\gamma)$.

Now since both $\ell^*$ and $\gamma$ are solutions of the equation $\sum_{i=1}^\numclass \frac{\lambda_i g_i(\ell)}{s_i(g_i(\ell))}\E{X_i} = n,$ which only has a unique solution, we have that $\gamma=\ell^*$. Thus the optimal $k_i=g_i(\ell^*)$, and the optimal value of the optimization problem~\eqref{eq:online opt} is \[ \frac{1}{\lambda}\sum_{i=1}^\numclass \frac{c_i\lambda_i\E{X_i}}{s_i(g_i(\ell^*))}.\]

\end{proof}

\section{Simulation Details}
\label{app:sim}

This appendix provides the full workload specification for Section~\ref{sec:simulation} and reports additional experiments that vary the inter-arrival time and job size distributions. As discussed in Section~\ref{sec:simulation}, the qualitative message is the same across all settings we tested: \policyone and \policytwo converge to the lower bound, the heuristics do not, and \policytwo beats the heuristics even when the number of cores is small. The figure shown in the main text is one representative experiment.

\paragraph{Workload used in the main text}

\noindent\textbf{Class 1:} sub-linear power speedup $s_1(k)=k^{0.9}$, unit mean job size ($\E{X_1}=1$), weight $c_1=1$;

\noindent\textbf{Class 2:} Amdahl-law speedup $s_2(k)=\frac{k}{\alpha k + (1-\alpha)}$ with serial fraction $\alpha=0.2$ (so the asymptotic speedup is capped at $1/\alpha=5$), unit mean job size ($\E{X_2}=1$), weight $c_2=1$.

Each class carries load $\rho_i=24$ on a base of $80$ cores (so the system load is $\rho = (\rho_1+\rho_2)/80 = 0.6$). As the number of cores $n$ scales, the arrival rate of each class scales proportionally to keep $\rho$ fixed. The main-text figure (Figure~\ref{fig:simulation_results}) uses Poisson arrivals and Exponential job sizes (SCV $=1$).

\paragraph{Additional inter-arrival time and job size distributions} To illustrate that the picture in Figure~\ref{fig:simulation_results}(b) does not depend on the choice of distributions, we repeat the heuristic comparison under two further regimes that fix the workload above but change the variability:
\begin{enumerate}
    \item \textbf{erlang:} Erlang-2 inter-arrival times and Erlang-2 job sizes, a lower-variability regime (SCV $=0.5$).
    \item \textbf{hyperexp:} Balanced two-phase Hyper-Exponential inter-arrival times and job sizes with SCV $=10$, a high-variability regime.
\end{enumerate}

The results are shown in Figure~\ref{fig:simulation_results_gen}. The qualitative behavior matches the exp regime reported in Figure~\ref{fig:simulation_results}(b): \policytwo consistently outperforms both heuristics at every $n$ and tracks the lower bound as $n$ grows, while the heuristics do not converge to it.

\begin{figure}[t]
    \centering
    \begin{minipage}{0.4\textwidth}
        \centering
        \includegraphics[width=\textwidth]{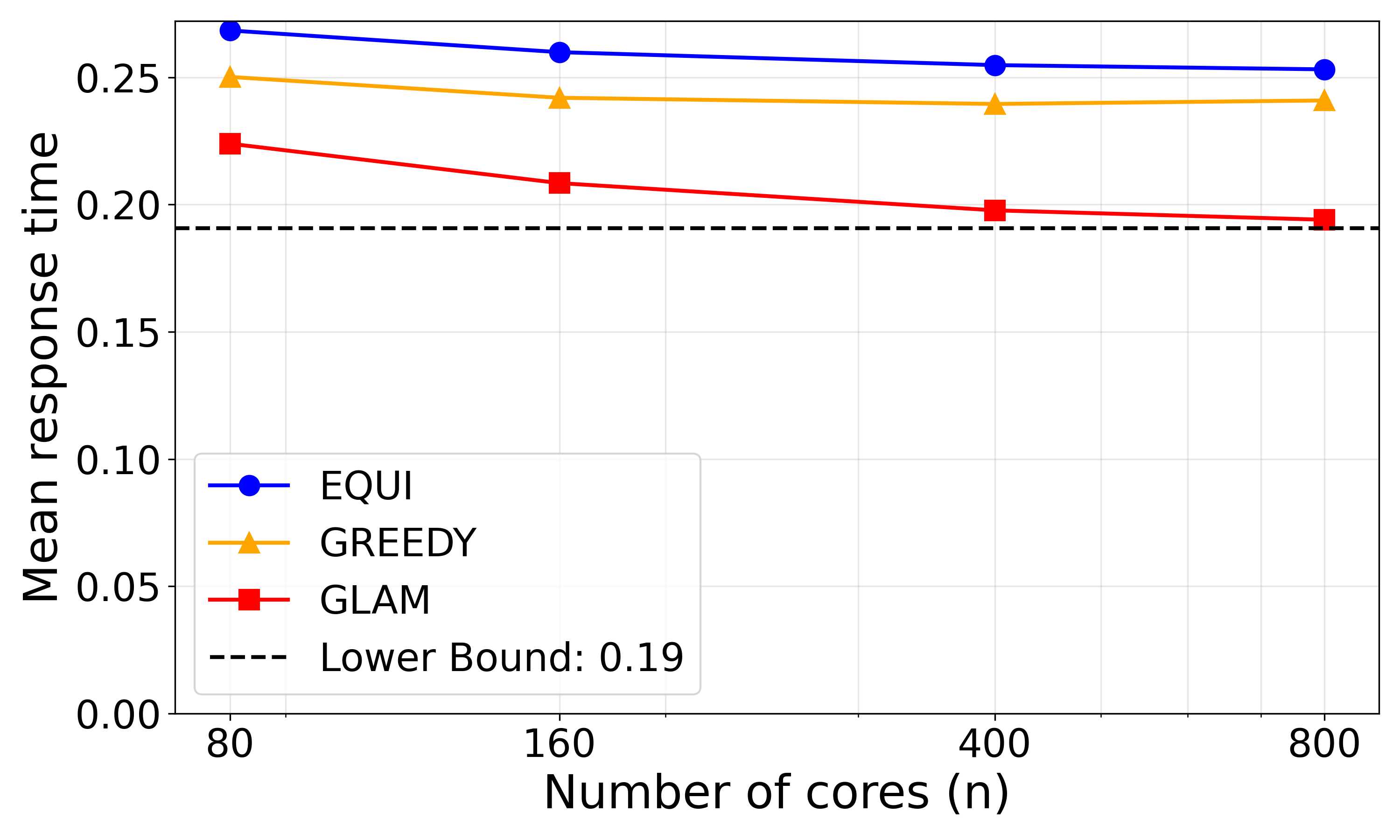}
        \\[2pt]
        (a) Erlang (SCV $=0.5$)
    \end{minipage}
    \qquad
    \begin{minipage}{0.4\textwidth}
        \centering
        \includegraphics[width=\textwidth]{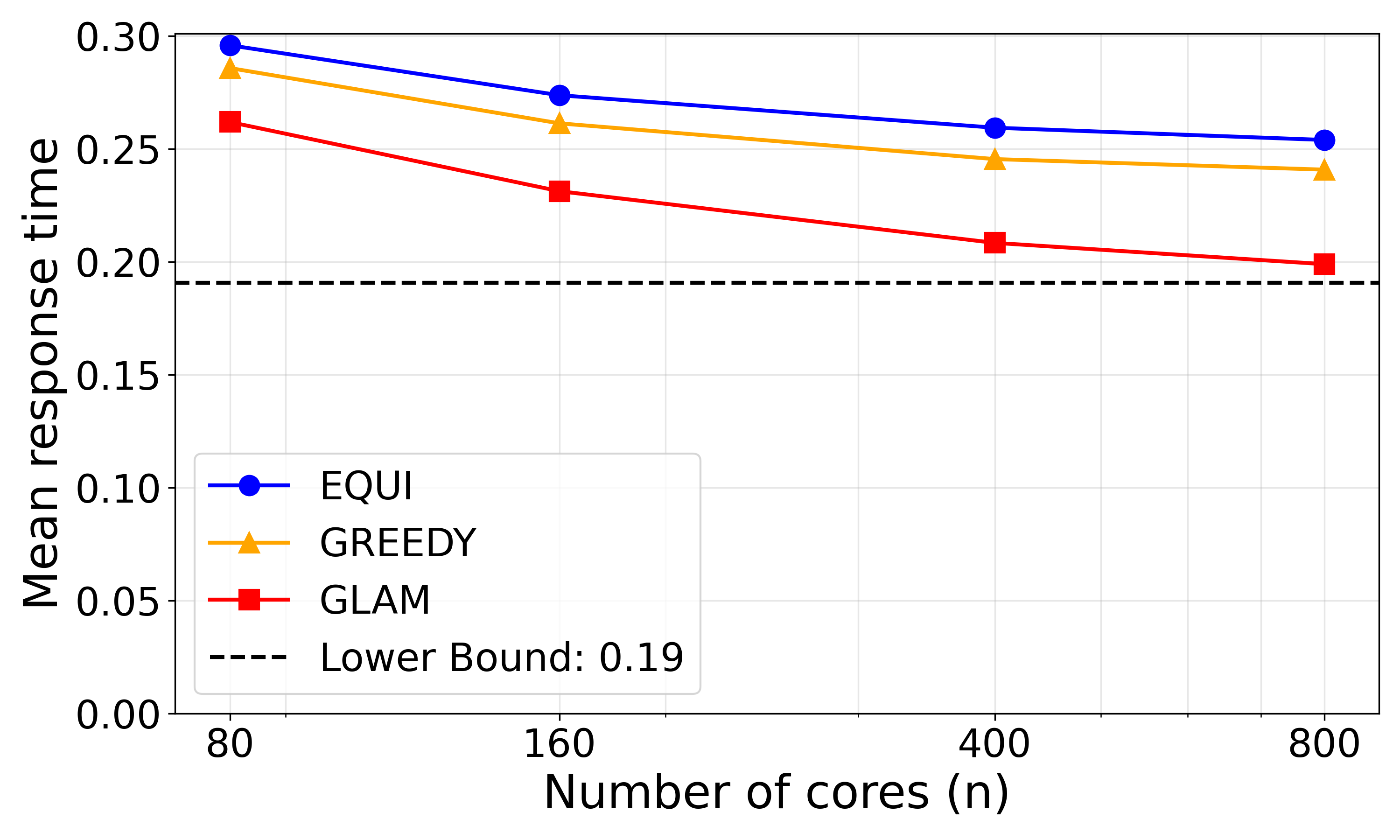}
        \\[2pt]
        (b) hyperexponential (SCV $=10$)
    \end{minipage}
    \caption{\policytwo vs. EQUI and GREEDY under inter-arrival and job size distributions with different variability. (a) Erlang-2 inter-arrival times and Erlang-2 job sizes (SCV $=0.5$). (b) Balanced two-phase Hyper-Exponential inter-arrival times and job sizes (SCV $=10$). In both cases \policytwo consistently outperforms EQUI and GREEDY and tracks the lower bound, while the heuristics do not converge to it. System load is fixed at $\rho=0.6$.}
    \label{fig:simulation_results_gen}
\end{figure}

\end{document}